%% file: main.tex
\documentclass[journal,transmag]{IEEEtran}

\usepackage{enumitem}
\usepackage{tagging}

\usetag{trp} 

\usepackage{fullpage}
\usepackage[utf8]{inputenc}
\usepackage{soul}
\usepackage{makecell}
\usepackage{glossaries}
\usepackage{graphicx}
\usepackage{xcolor}
\usepackage{multirow}
\usepackage{listings}
\usepackage{cleveref}
\usepackage{subcaption}
\usepackage{algorithm,algorithmicx}
\usepackage[noend]{algpseudocode}
\usepackage{amsmath}
\usepackage{amsthm}
\usepackage{array}
\usepackage{esvect}
\usepackage{tabularx}
\usepackage{adjustbox}
\usepackage{booktabs}
\usepackage{threeparttable}
\usepackage{pgfplots}
\usetikzlibrary{spy}

\definecolor{eclipseBlue}{RGB}{42,0.0,255}
\definecolor{eclipseGreen}{RGB}{63,127,95}
\definecolor{eclipsePurple}{RGB}{127,0,85}
\definecolor{pythonOrange}{RGB}{227,98,9}
\definecolor{pythonPurple}{RGB}{111,66,193}

\crefformat{figure}{Fig.~#2#1#3}
\crefformat{section}{Sec.~#2#1#3}
\crefformat{table}{Tbl.~#2#1#3}

\usepackage[colorinlistoftodos,textsize=tiny,textwidth=10mm]{todonotes}

\newtheorem{prop}{Proposition}

\newcommand{\en}{endorsement\ }
\newcommand{\ed}{endorser\ }

\title{Optimal endorsement for network-wide distributed blockchains}
\author{
    \IEEEauthorblockN{Iman Lotfimahyari, Paolo Giaccone}
    
	\IEEEauthorblockA{Dipartimento di Elettronica e Telecomunicazioni - Politecnico di Torino - Torino, Italy\\
    e-mail: firstname.lastname@polito.it}
}

\begin{document}

\maketitle

\begin{abstract}
%

Blockchains offer trust and immutability in non-trusted environments, but most are not fast enough for latency-sensitive applications. Hyperledger Fabric (HF) is a common enterprise-level platform that is being offered as Blockchain-as-a-Service (BaaS) by cloud providers. 
In HF, every new transaction requires a preliminary {\em \en }by multiple mutually untrusted parties called organizations, which contributes to the delay in storing the transaction in the blockchain.  
The {\em \en policy} is specific to each application and 
defines the required approvals by the {\em \ed peers} ({\bf EP}s) of the involved organizations.

In this paper, given an input {\em \en policy}, we studied the optimal choice to distribute the \en requests to the proper {\bf EP}s. We proposed the {OPEN} algorithm, devised to minimize the latency due to both network delays and the processing times at the {\bf EP}s. 
By extensive simulations, we showed that OPEN can reduce the {\em \en latency} up to 70\% compared to the state-of-the-art solution and approximated well the introduced optimal policies while offering a negligible implementation overhead compared to them.

%
%

\end{abstract}


\tagged{trp,jrn}
\input{Sec1_introduction.tex}

\tagged{trp,jrn}
\input{Sec2_EndPolHF.tex}

\tagged{trp,jrn}
\input{Sec3_AnalyticalModel.tex}

\tagged{trp,jrn}
\input{Sec4_EndOptimization.tex}

\tagged{trp,jrn}
\input{Sec5_NumEvaluation.tex}

\tagged{trp,jrn}
\input{Sec6_RelWorks.tex}
\tagged{trp,jrn}
\input{Sec7_Conclusions.tex}

\bibliographystyle{IEEEtran}
\bibliography{references}

\begin{IEEEbiography}[{\includegraphics[width=1in,height=1.25in,clip,keepaspectratio]{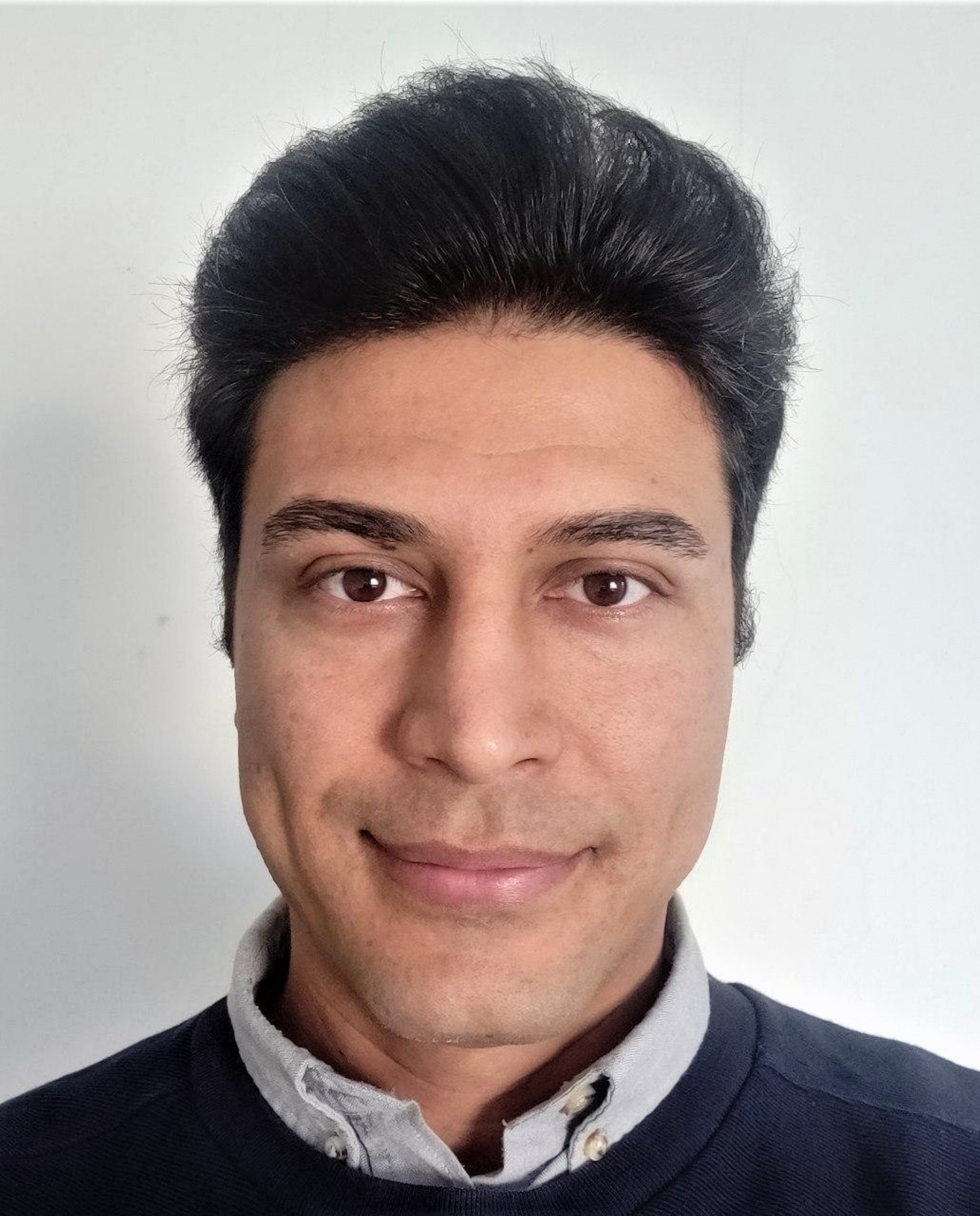}}]
    {Iman Lotfimahyari} received his BSc and MSc in Electronics Engineering from IAU University in 2003 and 2007 respectively. In March 2020, he received his second MSc in Telecommunication Engineering from Politecnico di Torino, Italy, and joined the Telecommunication Networks Group at Politecnico di Torino as a Ph.D. student. His current research interests involve programmable data planes for SDN and Blockchains.   
    \end{IEEEbiography}

   \begin{IEEEbiography}[{\includegraphics[width=1in,height=1.25in,clip,keepaspectratio]{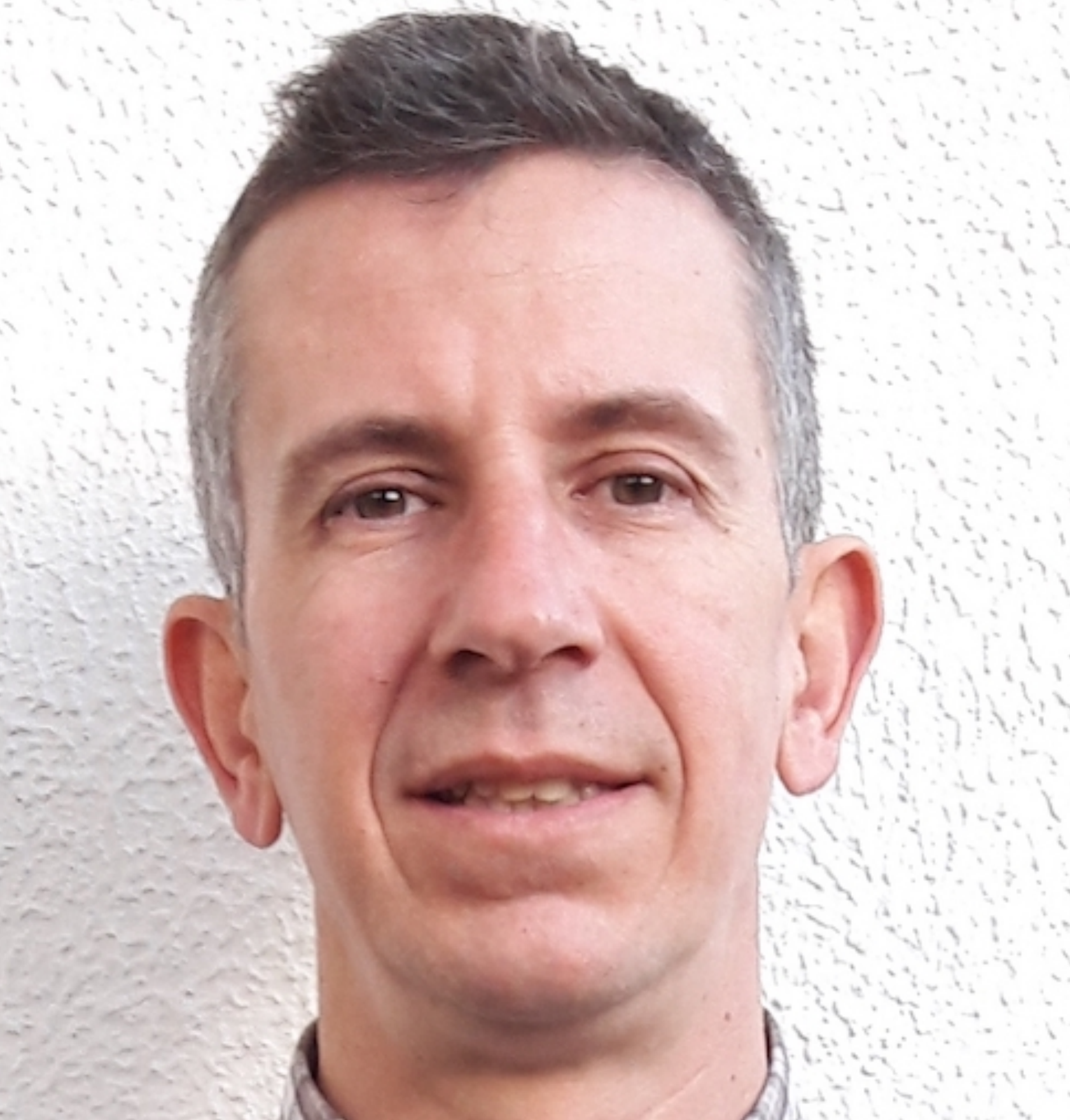}}]
    {Paolo Giaccone} received the Dr.Ing. and Ph.D. degrees in telecommunications engineering from the Politecnico di Torino, Italy, in 1998 and 2001, respectively. He is currently an Associate Professor in the Department of Electronics, Politecnico di Torino.
    During 2000-2001 and 2002, he was with the Information Systems Networking Lab, Electrical Engineering Dept., Stanford University, Stanford, CA. His main area of interest is the design of network control and optimization algorithms.
    \end{IEEEbiography}
 

\end{document}

%% file: Sec1_introduction.tex
\section{Introduction}

Nowadays, blockchains have become more and more relevant in many ICT applications. Pioneered by Bitcoin~\cite{bitcoin} and Ethereum~\cite{eth}, blockchains bring trust between different entities where trust is either nonexistent or unproven. They can improve security and privacy while offering a decentralized structure. The provided immutability brings visibility and traceability, beneficial for ICT applications, such as banking, supply-chain, IoT, healthcare, and energy sectors~\cite{bc1,bc2,bc3}.  

A blockchain is public if it is open to everyone to read otherwise it is private. But, if a node needs permission to participate in validating transactions, then the blockchain is permissioned otherwise it is permissionless~\cite{hf}. In contrast to public permissionless blockchains like Bitcoin, many enterprise applications require performance that permissionless blockchains are unable to deliver. Furthermore, many use cases necessitate knowing the identity of the participants, such as in financial transactions where notary service regulations must be followed. Private permissioned blockchains, such as Hyperledger Fabric (HF)~\cite{hf1} and Corda~\cite{corda}, meet such requirements.
In HF, a transaction must be endorsed (i.e., approved) by the organizations constituting the blockchain, according to a specified \en policy. This guarantees a mutual agreement between non-trusted parties, similarly, in the physical world, to a receipt declaring an asset transfer between two parties, signed by both parties.

HF uses an architecture called Execute-Order-Validate for transactions, enabling the definition of \en policies. During the execution phase, the client sends the transaction to some Endorser Peers (EPs), based on the user's specified \en policy. Each EP processes the transaction by only simulating it without applying the results on the blockchain. The simulation result, denoted as ``endorsement", is signed by the EP and returned to the client. Finally, if the \en policy is satisfied, the signed and endorsed proposal of the transaction will be sent to the blockchain nodes to be stored.

The \en delay experienced by a client is affected mainly by two components: i) the network delay between the client and the EPs and ii) the processing delay at each EP. 
The network delay mainly depends on the network congestion and the propagation delays, whereas the processing delay depends on both the CPU capability and the computation load of each EP.
Because network and processing delays are time-varying and hence difficult to predict, optimally selecting EPs is hard. Note that a selection algorithm choosing just the best EP based on the minimum experienced delays will concentrate the \en requests to the same EPs, increasing the network congestion and the processing load, thus increasing the overall \en delays. 
In this work, we propose an {\em optimal EP selection policy} minimizing the \en delays. The main idea is to send redundant \en requests to multiple EPs. The adopted spatial diversity increases the chance of having the best EPs among the selected ones. The benefit of the proposed approach can be captured by a simple queueing model in which a task is sent in parallel to multiple servers, each with its queueing system, to minimize task completion time.

In this paper, our novel contributions are as follows:
\begin{itemize}
 \item We highlight the role of the network and processing delays in the overall \en delay.
 \item We refer to a simple analytical model, based on the classical theory of queueing systems, to evaluate the effect of redundancy in selecting the EPs and to compute the optimal number of EPs. 
 \item We propose an optimization approach denoted as \textit{OPtimal ENdorsement~(OPEN)} based on the analytical results, leveraging the history of endorsement delays.  
\item We demonstrate through extensive simulations that OPEN outperforms the state-of-the-art solution and accurately approximates other optimal policies while having a much lower implementation overhead compared to them.
\item We provide some preliminary results regarding a proof-of-concept implementation.
\end{itemize}

The rest of this paper is structured as follows. Sec.~\ref{sec:hf} describes the HF architecture and then focuses on the endorsement phase delay by introducing the network model and the EP selection problem. Sec.~\ref{sec:mod} explains a simple analytical model, derived from classical results on queueing theory, to find the optimal number of EPs in a simplified scenario. In Sec.~\ref{sec:opt}, we propose an EP selection algorithm based on the optimal replication factor computed analytically in Sec.~\ref{sec:mod}, able to operate in a generic scenario. In Sec.~\ref{sec:simu}, we assess by simulation the performance of our proposed approach and compare it with the alternatives proposed and with the state-of-the-art solution.
In Sec.~\ref{sec:poc} we show some preliminary results of a proof-of-concept implementation of OPEN. 
 In Sec.~\ref{sec:related} we discuss the related work. Finally, we draw our conclusions in Sec.~\ref{sec:conc}.

%% file: Sec2_EndPolHF.tex
\section{Hyperledger Fabric architecture and \en}\label{sec:hf}

\tagged{trp}{
We start by describing the HF architecture and then we focus on the \en phase.
\subsection{Hyperledger Fabric architecture and protocol}
The Execute-Order-Validate approach enables the simulation of the transactions before the agreement of the participants on recording the results in the blockchain.
We highlight the role of the following entities in the reported reference architecture in~\cref{fig:arch}.
\begin{figure}[!tb]
  \centering
  \includegraphics[width=7cm]{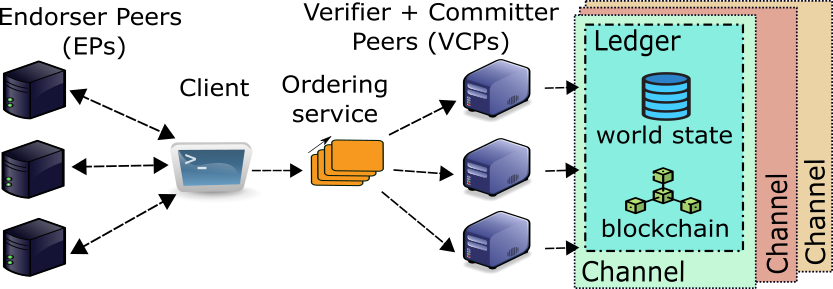}
  \caption{The main entities involved in HF architecture}
  \label{fig:arch}
\end{figure}
}
\tagged{jrn}{\subsection{Hyperledger Fabric architecture and protocol}
The Execute-Order-Validate approach enables the simulation of the transactions before the agreement of the participants on recording the results in the Hyperledger Fabric blockchain.
We describe the role of the entities which are participating in the simulation phase of~\cref{fig:phase}.
\begin{figure}[!tb]
  \centering
  \includegraphics[width=8cm,height=4.5cm]{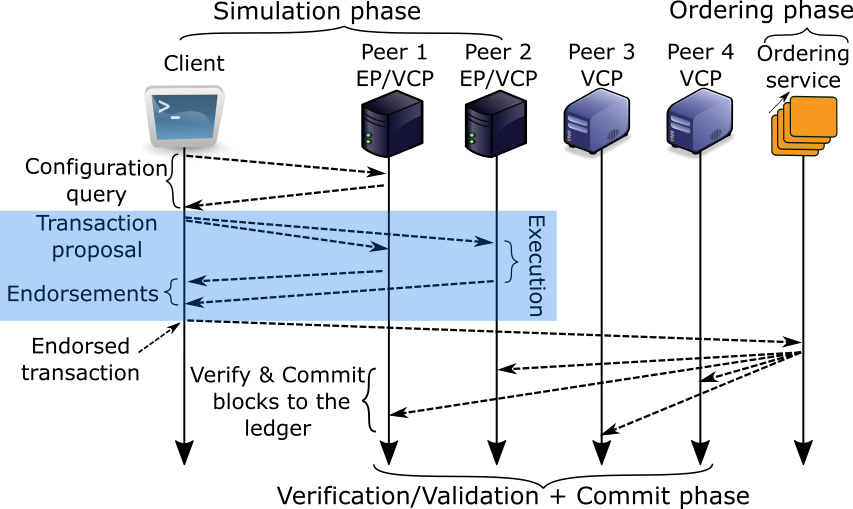}
  \caption{{Transaction processing phases in HF highlighting all the message interactions between the involved entities}
  }
  \label{fig:phase}
\end{figure}
}

The \textit{client} 
is responsible for preparing the transaction proposal of the users' transactions and sending it to the Endorser Peers (defined below) based on the specified \en policy. If the client receives enough endorsements before a specific time-out, it forwards them to the ordering service; otherwise, the client can re-transmit the same proposal in the hope of receiving enough endorsements in time. 

\tagged{trp}{The \textit{ledger} stores transactions in a distributed manner and is composed of two distinct yet related parts. The first is the actual ``blockchain," which maintains a log of transactions shared among all nodes and is immutable by design.
The other is the ``world state" which holds the most recent state associated with the validated transactions, which is recorded on the blockchain log.} 

The \textit{Peer} is the element responsible for the following tasks. The \textit{Endorser Peer}~(EP) simulates/executes the transaction received from the client application, based on the current values of the world state. The \textit{Verifier/Committer}~(VCP) receives a block of simulated transactions from the ordering service and verifies their legitimacy to mark them as validated or invalidated. Then it appends the verified block to the blockchain, comprising all the transactions (validated or invalidated).
To update its copy of the ledger, an EP is typically a VCP at the same time. The peers are owned by various \textit{organizations} that are blockchain members. An organization can be as small as individuals or as large as a multi-national corporation.
\tagged{trp}{The \textit{ordering service} receives endorsed transactions from clients, aggregates them into blocks, and distributes them to the VCPs.
The \textit{channel} is a private blockchain overlay that provides data confidentiality and ledger isolation. The transacting parties must be authenticated on a channel to read/write the corresponding data.

The \textit{smart contract} is a piece of code responsible for validating the transactions and thus performing any read-set/write-set of interactions with the ledger. Its functions include simulating the transactions, validating them based on the \en policy and corresponding word state, and finally updating the world state based on validated transactions.}

The \textit{\en policy} defines the logical conditions to validate a transaction in terms of the EPs on a channel that must execute a transaction proposal. In Sec.~\ref{sec:endo}, we will describe in detail the representation of the \en policy.
The definition of an \en policy is at the organizational level, which means any EP of that organization can represent that organization in the \en policy. A transaction should pass three phases to be stored in the blockchain, as shown in~\cref{fig:phase}.

\tagged{trp}{In the \textit{simulation phase} the client prepares the proposal for the transaction and sends it to a set of EPs, which depends on the specified \en policy. The EPs execute the transaction proposal and return the signed results to the client. As soon as the \en policy is satisfied, the client sends the signed \en results to the ordering service.
During the \textit{ordering phase} the endorsed transactions received from the clients are ordered and packed to create new blocks. The VCPs will receive these blocks.

In the \textit{verification/validation and commit phase}, the received blocks from the ordering service are verified to have legitimate transactions based on \en requirements and world state values. 
The block will be added to the blockchain by the VCPs, and the world state will be updated using the validated transactions.

\begin{figure}[!tb]
  \centering
  \includegraphics[width=8cm]{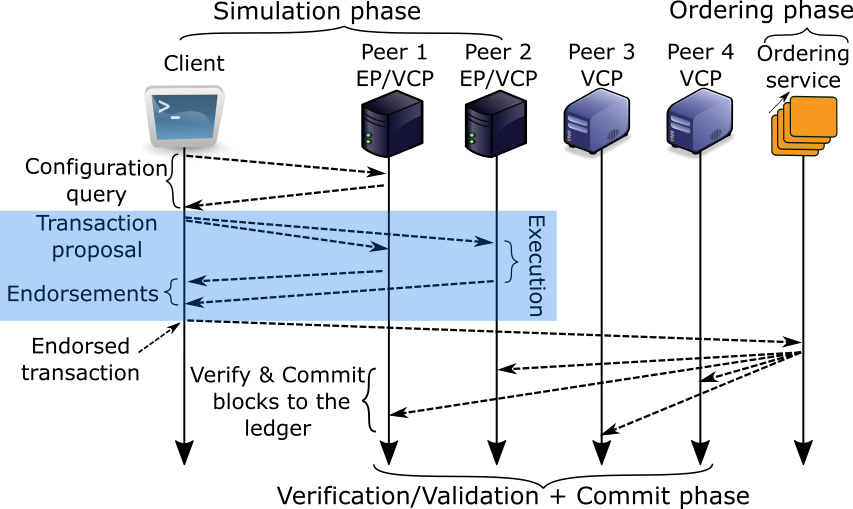}
  \caption{{Transaction processing phases in HF highlighting all the message interactions between the involved entities}
  }
  \label{fig:phase}
\end{figure}}

\subsection{Standard form of an \en policy}\label{sec:endo}

HF provides a very flexible way to define an \en policy.
We will show that any \en policy, despite its complexity, can be reduced to a standard form. 
In HF, the definition of an \en policy is based on a syntax that allows the operators ``AND", ``OR" and ``$k$-OutOf" to be applied to a set of organizations and nested expressions~\cite{hfend}. In particular, the operator ``$k$-OutOf-$E$" returns true whenever at least $k$ expressions within set $E$ are satisfied. Despite the complexity of the policy expression, we prove that the following proposition holds:
\begin{prop}\label{pr:p1}
Any \en policy obtained by combining arbitrarily ``AND", ``OR", and ``OutOf" operators is equivalent to the policy: 
 \begin{equation}\label{eq:or1}
    \text{OR}(St_1,St_2,...)
 \end{equation}
where each $St_i$ is either a single organization or the conjunction~(``AND") of different organizations.
\end{prop}
\begin{proof}
In the case of expressions based on only ``AND" and ``OR" operators, thanks to the distribution principle in logic expressions, we can transform the original expression into the target form~\eqref{eq:or1}. In the case of ``$k$-OutOf$(e_1,e_2,\ldots, e_m)$ operator, where $e_i$ is a single expression, by definition this holds:
\begin{equation*}\label{eq:kout1}
  k\text{-OutOf}(e_1,e_2,\ldots, e_m)=\text{OR}(\{\text{AND}(E)\}_{E\in\Omega})
\end{equation*}
being $\Omega$ the set of all $\binom{m}{k}$ combinations of $k$ expressions from the set of $m$. 
Now, since any expression with the ``OutOf" operator is equivalent to one with only ``AND" and ``OR", by following the previous reasoning, such expression can be reduced to the expression~\eqref{eq:or1}. 
\end{proof}

The policy, defined at the organization level, must be mapped into a policy defined at the EP level since the \en requests should be sent to the proper EPs. 
So, getting the \en from a specific organization requires receiving it from {\em any} of its EPs, which is equivalent to the policy $1$-OutOf$(p_1,p_2,\ldots)$, where $p_i$ are the EPs within the organization. Revisiting Proposition~\ref{pr:p1} applied at the policy expression at the EP level, we can claim:
\begin{prop}\label{pr:p2}
Any \en policy defined at the organization level can be expanded into an \en policy defined at the EP level as follows:
\begin{equation}\label{eq:or2}
\text{OR}(St^{'}_1,St^{'}_2,...)
\end{equation}
where each of $St^{'}_i$ is either a single EP or the conjunction~(``AND") of different EPs.
\end{prop}

The result of Proposition~\ref{pr:p1} allows investigating only one standard form of \en expression, independently from the original expression complexity. Now, by using Proposition~\ref{pr:p2}, we will have the \en expression extended at the EP level. At this level, the final \en will be just in the form of the OR between the conjunction~(``AND") of different EPs of different organizations, as in~\eqref{eq:or2}. 

For example, consider a scenario with three organizations and two EPs in each of them. If the \en policy is ``$2$-OutOf$(o_1,o_2,o_3)$", then we can rewrite it as:
\begin{multline}\label{eq:2out2}
    \text{2-OutOf}(o_1,o_2,o_3)=\\
    \text{OR}(\{\text{AND}(p_{ij},p_{i^{'}j^{'}}), \forall i, \forall i'\neq i, \forall j, \forall j'\}) 
\end{multline}
where $o_i$ is \text{organization} $i$, and $p_{ij}$ is the EP $j$ of \text{organization} $i$. The expanded version in~\eqref{eq:2out2} lists all the possible combinations of the EPs that can satisfy the \en policy according to the standard form.



\subsection{Endorser peer (EP) selection algorithm}

In our work we focus on the EP selection algorithm, starting from the standard form of the \en policy. 
%
The {\em \en delay} is the amount of time the client waits, from sending the \en request until receiving the first \en reply that satisfies the \en policy. 
The response delay from an EP is the sum of two components: the network delay and the processing delay at the EP. The {\em network delay} depends on the propagation delay and the queueing delay along the path to the EP, which is affected by the time-variant congestion conditions. The overall {\em processing delay} depends on the {\em queuing} at the EP before being served and the {\em computation time} at the EP, which depends on the CPU speed and the instantaneous CPU load and resource contentions.

Because the standard form of any \en policy comprises an overall ``OR" operator, as in~\eqref{eq:or2}, the \en latency corresponds to the {\em minimum} delay to get a valid statement. Also, each statement is based on an ``AND" operator between EPs, so the delay of each statement depends on the {\em maximum}  response delay of all EPs included in a statement. 
In summary, the \en latency depends on the ``fastest'' group of EPs forming a statement, while the delay of each group depends on the ``slowest" EP within the group.

\subsection{System model for the \en phase}

Without loss of generality, we consider a fixed network topology connecting $C$ clients with $Q$ organizations, each of them with a generic network connecting the internal EPs, depicted in Fig.~\ref{fig:model1}. We assume that all nodes in the system are always available, the routing is fixed, and the links have enough bandwidth to prevent network congestion caused by the \en protocol. Thanks to the service discovery process in HF, we consider only the most updated EPs.

%% file: Sec3_AnalyticalModel.tex
\section{Background on optimal replication in queueing systems}\label{sec:mod}

Now, we discuss an analytical model to compute the optimal number of EPs for each transaction, derived from classical results on task replication in a queueing system, as explained in Sec.~\ref{sec:related}. 
For the sake of readability, we report the adopted notation in Table~\ref{table:notations}.

\begin{figure}[!tb]
\minipage{0.5\columnwidth}
\centering
  \includegraphics[width=0.9\textwidth,height=2.48cm]{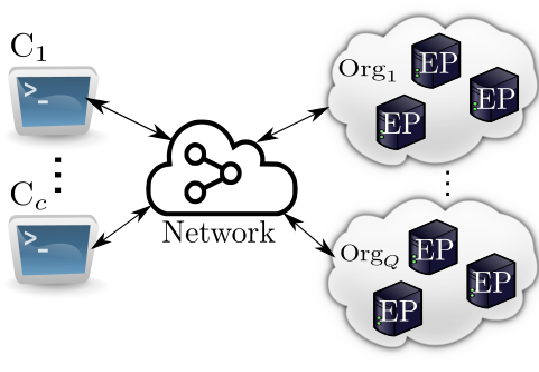}
  \subcaption{General network model}
  \label{fig:model1}
\endminipage
\minipage{0.5\columnwidth}
\centering
  \includegraphics[width=0.9\textwidth]{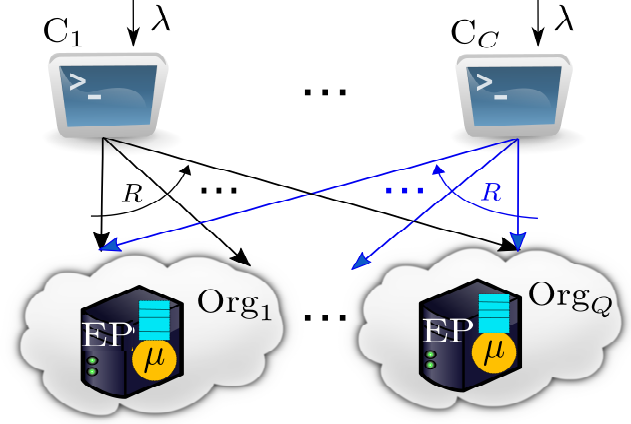}
  \subcaption{$R$ requests for $1$-OutOf-$Q$}
  \label{figure:SmplSysModel}
\endminipage
\caption[width=1\linewidth]{{Network model and the \en policy sending each request to $R$ organizations/peers in parallel.}}
\label{figure:NetmodelPolicy}
\end{figure}

\begin{table}[!tb]
\caption{Notation}
\begin{tabularx}{1.0\columnwidth}{@{}p{0.10\columnwidth}p{0.85\columnwidth}@{}}
\toprule
  $C$ & number of clients\\
  $Q$ & number of organizations\\
  $p$ & endorser peer (EP)\\
  $\lambda$ & arrival rate of new transactions to each client\\
  $\mu$ & inverse of computation time for the EP server\\
  $U$ & utilization factor in each EP server \\
  $R$ & redundancy factor \\ 
  $W_i$ & waiting time needed for the $i$th request to be served\\
  $S_i$ & inter-arrival time between the ordered version of  $\{W_i\}_i$\\
  $\gamma$ & normalized load factor for the worst case $R=Q$\\
  $\hat{R}_k$ & optimal $R$ for policy $k\text{-OutOf-}Q$\\
  $L_k$ & endorsement latency for policy $k\text{-OutOf-}Q$\\
  $\mathcal P$ & set of all available EPs\\
  $\mathcal P_e$ & set of selected EPs\\
  $\mathcal P_e^{\text{old}}$ & set of previously selected EPs\\
  $T$ & probe sampling period\\
  $\text{TX}^n$ & transaction with local sequence number $n$\\
  $x^k_p$ &  \en latency of $\text{TX}^{k}$ for peer $p$\\
  $t^\text{resp}_p$ & virtual response delay of a peer for the new TX\\
  $t^\text{busy}_p$ & virtual time at which  EP $p$ is not busy anymore\\
  $\tau_p^\text{proc}$ & processing delay for the current endorsement request\\
  $\tau^\text{net}_p$ & network delay for EP $p$\\
  $\tau^\text{queue}_p$ & queueing time experienced by the TX at EP $p$\\
  $d_{cp}$ & the network delay between client $c$ and EP $p$\\
\bottomrule
\end{tabularx}
\label{table:notations}
\end{table}

We consider a simplified model as shown in~\cref{figure:SmplSysModel}, with one EP in each organization. We assume $1\text{-OutOf-}Q$ as the \en policy, which corresponds to $\text{OR}(p_1,p_2,\ldots,p_Q)$ in its standard form. For now, we neglect the network delays and concentrate just on processing delays.
We suppose each client generates \en requests according to a Poisson process with rate $\lambda$. Each client selects at random $R$ EPs to send the \en request. $R$ will be denoted in the following as {\em redundancy factor}. To model the processing time variability at the EP, we assume an exponentially distributed processing time with an average $1/\mu$, coherently with past works~\cite{end6,end3}. Thus each EP can be modeled as an M/M/1\footnote{In classical queueing theory, an M/M/1 queue has a single server,  arrivals follow a Poisson process and service times are exponentially distributed~\cite{mm1}.} queue with arrival rate {$\lambda RC/Q$} and service rate $\mu$. We define the utilization factor for each EP as {$ U=\lambda RC/\mu Q$}.
 Thus, for the request traffic to be sustainable, {$U < 1$} and the \en request arrival rate must satisfy {$\lambda < \mu Q/RC$}.
We can now claim the following:
\begin{prop}\label{lm:l1}
Under a sustainable arrival rate of \en requests and a random selection policy with $R$ EPs, according to the \en policy $1\text{-OutOf-}Q$, it holds for the \en latency $L_1$:
\begin{equation}\label{eq:1latency0}
    E[L_1]=\dfrac{1}{\mu-\dfrac{\lambda RC}{Q}}\left(\dfrac{1}{R}\right)\qquad R\in[1,\ldots,Q]
\end{equation}
\end{prop}
\begin{proof}\label{pr:pl1}
From~\cref{figure:SmplSysModel}, let $\lambda'$ be the average incoming rate of the requests for the queue of each EP such that: 
\begin{equation}\label{eq:1latency1}
    \lambda^{'}=\dfrac{\lambda RC}{Q}
\end{equation}
We define $W_i$ as the waiting time of a request to be served at the $i$th EP, which is the sum of queuing time and the serving time of the request in the $i$th EP. From M/M/1 well-known properties~\cite{mm1}, $W_i$ are i.i.d. and exponentially distributed with mean: 
{$E[W_i]=1/(\mu-\lambda^{'})$}.

Observe that: {$\enskip L_1=\text{min}(W_1,W_2,\ldots,W_R)$}
where $W_i$ are i.i.d..
From basic properties of the exponential distribution, $L_1$ is exponentially distributed with mean:
\begin{equation}\label{eq:1latency6}
    E[L_1]=\dfrac{E[W_i]}{R}
\end{equation}
and finally get \eqref{eq:1latency0}.
\end{proof}

By computing the first derivative of~\eqref{eq:1latency0} with respect to $R$, we can prove the following:
\begin{prop}
Let $\hat{R}_1$ be the optimal value of $R$ that minimizes $E[L_1]$ for the policy $1\text{-OutOf-}Q$. 
\begin{equation}\label{eq:1rhat1}
    \hat{R}_{1}=\dfrac{\mu Q}{2\lambda C}
\end{equation}
\end{prop}
In summary, the optimal number of EPs changes with $\lambda$. For low arrival rates, $R$ must be large to exploit the spatial diversity, without incurring additional overhead in the processing times. For high arrival rates, conversely, $R$ is small to reduce the load on the EPs. Notably, for the sake of readability, we omitted from~\eqref{eq:1rhat1} the clipping to the interval $[1, Q]$ and the rounding procedure to find the optimal integer value of $R$.
We can now extend the result of Proposition~\ref{lm:l1} to a generic OutOf policy.
\begin{prop}
Under a sustainable arrival rate of \en requests and a random selection policy with $R$ EPs, according to the \en policy $k\text{-OutOf-}Q$, it holds for the endorsement latency $L_k$:
\begin{multline}\label{eq:2latency0}
    E[L_k]=\dfrac{1}{\mu-\dfrac{\lambda CR}{Q}}\left(\sum_{i=0}^{k-1}\dfrac{1}{R-i}\right)\qquad\\
    R\in[k,\ldots,Q]
\end{multline}
\end{prop}

\begin{proof}
Using the same definition of $W_i$ as adopted in the proof of Proposition~\ref{lm:l1}, we can define 
$L_k$ as the \en latency for the policy $k\text{-OutOf-}Q$. Now $L_k$ can be computed as the $k$th order statistic as follows, 
{$L_k=(W_1,W_2,\ldots,W_R)_{(k)}$}, recalling the fact that $W_i$ are i.i.d.\ and exponentially distributed, we can define $S_i$ as the time interval between the ordered version of the $W_i$ (i.e., $S_i=W_{(i+1)}-W_{(i)}$).
Thanks to the theory of order statistics~\cite{order}, $S_i$ is exponentially distributed with average:
\begin{equation}\label{eq:2latency3}
    E[S_i]=\frac{E[W_i]}{(R-i)}
\end{equation}
{By combining~\eqref{eq:1latency6} and~\eqref{eq:2latency3} we have:
\begin{equation}\label{eq:2latency4}
    E[L_k]=\sum_{i=1}^{k-1}E[S_i]+E(L_1)
\end{equation}
simplified to:
\begin{equation}\label{eq:2latency5}
    E[L_k]=\sum_{i=0}^{k-1}\frac{E[W_i]}{(R-i)}\qquad R\in[k,\ldots,Q]
\end{equation}
and we get~\eqref{eq:2latency0}.}
\end{proof}

The optimal value of $\hat{R}$ can be computed analytically as well.
We impose sustainable request arrivals, i.e., {$U < 1$}, {\em for any $R$} to guarantee sustainable arrivals 
also in the case $R=Q$, it must hold {$\lambda<\mu/C$}. Thus, we can set: 
\begin{equation}\label{eq:lambda3}
    \lambda=\gamma\dfrac{\mu}{C}
\end{equation}
with $\gamma\in(0,1)$ being the load factor. By substituting \eqref{eq:lambda3} into~\eqref{eq:1rhat1}, we can obtain the optimal number of EPs for $1$-Out-Of-$Q$ policy as:
\begin{equation}\label{eq:1rhat3}
    \hat{R}_{1}=\dfrac{Q}{2\gamma}
\end{equation}
We can repeat the same derivation also for $\hat{R}_k$, i.e., for a generic $k$-OutOf-$Q$ policy. 

\subsection{Numerical evaluation}

 In \cref{fig:Opt-R} we reported the \en latency computed in the function of $\gamma$ and $R$, obtained by 
 substituting~\eqref{eq:lambda3} in~\eqref{eq:2latency0}.
%
\begin{figure*}[!tb]
  \centering  \includegraphics[width=1\linewidth]{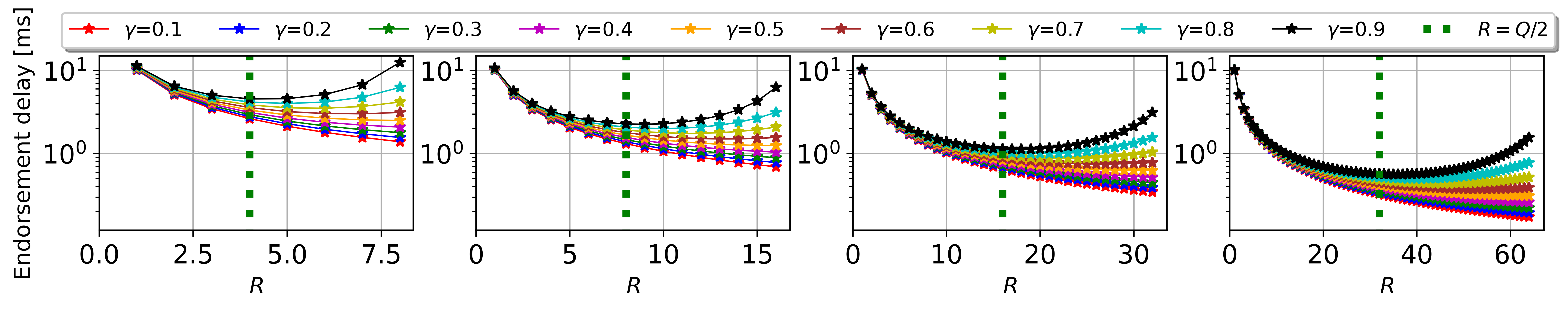}
  \caption{Endorsement latency where $Q \in [8,16,32,64]$ (left to right). The green line represents $R=Q/2$.
  }
  \label{fig:Opt-R}
\end{figure*}
As expected, we observe a minimum \en latency obtained with $R=R_{k}$, as computed analytically, which depends on the load $\gamma$.
Due to the difficulty to estimate the load in practical scenarios (which may not be stationary), for $k=1$, we propose heuristically choosing $R=Q/2$ as a sub-optimal redundancy in our proposed approach, discussed in the following. This choice is robust since it is optimal at high load and at low load the latency increase is limited. Indeed, for $\gamma=0.5$ the increase is no more than $8\%$ compared to the optimal value, and for $\gamma=0.1$ no more than $12\%$. Thus for $k=1$, $R=Q/2$ appears to be a practical solution, which will be exploited when devising online EP selection algorithms in Sec.~\ref{sec:opt}.


The redundancy effect can be limited due to the number of organizations/EPs or applied endorsement policies, as they affect the number of statements generated by Proposition~\ref{pr:p2}. With fewer final statements, there would be less space for redundancy. Indeed, systems with less restrictive and less complex endorsement policies (e.g., majority policies) benefit more from redundancy, while organizations benefit from adopting more  EPs to increase reliability. 

%% file: Sec4_EndOptimization.tex
\section{Practical endorsers selection algorithms}\label{sec:opt}

Now, we concentrate on the $1$-OutOf-$Q$ policy, since it is coherent with the standard form of any \en policy. Without loss of generality, we assume just one client in the system ($C=1$). We assume that the client is aware of all needed information including available/most-updated EPs, thanks to the configuration query request, as shown in Fig.~\ref{fig:phase}, which leverages the available service discovery process.

We propose an optimization procedure to select the EPs, denoted as \texttt{OPEN}, whose main goal is to minimize the \en response delay. \texttt{OPEN} considers the past response delays experienced by the previously selected EPs and selects the EPs with the lowest delays. 
This choice is motivated by the high temporal correlation between the response delays of an EP, due to queueing in the network and in  the EPs.
Notably, the history is meaningful only for recently selected EPs, otherwise, it is obsolete. Therefore, it is possible that a highly loaded EP which was not recently requested becomes among the least loaded ones and is worth again sending the request to it. To address this, \texttt{OPEN} probes non-selected EPs by sending gratuitous \en requests, which are still considered in the evaluation of the \en policy. 
Furthermore, in \texttt{OPEN} pending requests are considered indicators of possibly congested EPs, which are chosen at a lower priority. 



The pseudocode of \texttt{OPEN} is provided in Fig.~\ref{alg:OPEN}.
Let TX$^{n}$ be the transaction with sequence number $n$, evaluated locally at the client.
Let $x^n_p$ be the measured response delays of TX$^{n}$ for any EP $p\in\mathcal P$.
Let $\mathcal P_e^n$ be the set of selected EPs for TX$^n$. 
For each transaction, we initialize all EPs as eligible to be selected (ln.~\ref{lst:2}). Just for the first transaction, \texttt{OPEN} initializes the history of response delays to a dummy value and selects all EPs as selected endorsers (ln.~\ref{lst:3}-\ref{lst:6}). For a generic transaction, all response delays are initialized to a dummy value (ln.~\ref{lst:7}-\ref{lst:9}). Then the EPs are selected based on a procedure described in the next paragraph (ln.~\ref{lst:10}). 
Now \texttt{OPEN} sends the \en request for TX$^n$ to the computed set of EPs (ln.~\ref{lst:11}) and updates the measured delays (ln.~\ref{lst:12}). A new instance of the procedure would start if a new transaction TX$^{n+1}$ is generated. Note that the procedure ends when all the responses are received.

\begin{figure}[!tb]
\scriptsize
\begin{algorithmic}[1]
\renewcommand{\algorithmicrequire}{ Input:}
\renewcommand{\algorithmicensure}{ Output:}
\Procedure{OPEN}{$n$}\Comment{Process TX$^n$}
\State $e^{n}_p \gets$ {\bf true}\label{lst:2}, $\forall p\in\mathcal P$ \Comment{Init the eligibility vector for $\text{TX}^{(n)}$}
\If {$n=1$} \label{lst:3}\Comment{Just for the first transaction}
  \For{$p \in \mathcal P$}\label{lst:4}
    \State $x^{0}_p\gets x^{1}_p \gets -1$\label{lst:5} \Comment{Init the response delay history}
  \EndFor
  \State $\mathcal P_e^1\gets \mathcal P$ \label{lst:6} \Comment{Select all the available peers}
\Else \label{lst:7} \Comment{Consider a generic transaction}
\For{$p \in \mathcal P$}\label{lst:8}
  \State $x^{n}_p \gets -1$\label{lst:9} \Comment{Init the measured delays for $\text{TX}^{(n)}$}
\EndFor
\State $\mathcal P_e^n\gets$ {\tt Select-Endorsers}()\label{lst:10} 
\EndIf
\State {\tt Send-Endorsement-Requests}(TX$^n, \mathcal P_e^n$)\label{lst:11}
\State $X^n\gets$ {\tt Update-Response-Delays}()\label{lst:12}
\EndProcedure
\end{algorithmic}
\caption{Pseudocode of the \texttt{OPEN} algorithm for TX$^n$}\label{alg:OPEN}
\end{figure}

We now discuss how \texttt{Select-Endorsers} function operates.
Inspired by our previous result in~\eqref{eq:1rhat3}, it selects $|\mathcal P|/2$ EPs chosen among the ones that experienced the lowest response delays, based on the measures for the last transaction TX$^{n-1}$. The choice is challenging when one or more responses are still pending for TX$^{n-1}$, and the algorithm key idea is that the corresponding EPs are considered as congested and thus should not be selected for the current transaction TX${^n}$.

The pseudocode is reported in Fig.~\ref{alg:se}. 
It calculates the maximum delay measured for TX$^{n-1}$ (ln.~\ref{se2}). For each EP in $\mathcal P_e^{n-1}$ such that the response is not received yet, we mark the corresponding EP as non-eligible (ln.~\ref{se3}-\ref{se5}). There are two cases. The first case is the special one in which no responses have been received for TX$^{n-1}$, thus the algorithm speculates the delay equal to the delay of TX$^{n-2}$ (ln.~\ref{se6}-\ref{se7}). The eligibility assigned to the EPs will lead to selecting the other $|\mathcal P|/2$ EPs compared to the previous ones.
The second case is the typical one in which at least some responses have been received for TX$^{n-1}$ (ln.~\ref{se8}). For the EPs used in TX$^{n-1}$ and for which no response has been already received, the speculated delay is equal to the maximum delay $d_{\max}$ plus some constant $\epsilon$, chosen enough small to be negligible compared to the average network and processing delays (e.g., 1~ns) (ln.~\ref{se9}). This will model the fact that the actual delay is unknown, but for sure it is strictly larger than $d_{\max}$. Finally, for all the other EPs, not used for TX$^{n-1}$, the delays are speculated to be equal to $X^{n-2}$ (ln.~\ref{se10}-\ref{se11}). Now, the EPs are sorted based on the $X^{n-2}$ delay values and the half best will be selected (ln.~\ref{se12}). A random EP from not selected ones will be chosen as the gratuitous probe EP (ln.~\ref{se13}). The slowest EP from $\mathcal P_e^n$ will be replaced with the gratuitous probe EP (ln.~\ref{se14}), and $\mathcal P_e^n$ will be returned to the main OPEN process (ln.~\ref{se15}).

\begin{figure}[!tb]
\scriptsize
\begin{algorithmic}[1]
\Procedure{SelectEndorsers}{~}
  \State $d_{\max}=\max_{p\in\mathcal P_e^{n-1}}\{x_p^{n-1}\}$\label{se2} \Comment{Max measured delay  for TX$^{n-1}$}
  \For {$p \in \mathcal P_e^{n-1}$}\label{se3} \Comment{For EPs used for $\text{TX}^{n-1}$}
    \If {$x_p^{n-1}= -1$}\label{se4} \Comment{Not yet response from EP $p$}
      \State $e_p^{n} \gets$ {\bf false} \label{se5}\Comment{Make the EP Not-eligible for $\text{TX}^{n}$}
      \If{$d_{\max}=-1$}\label{se6}\Comment{No delay measured for $\text{TX}^{n-1}$}
        \State $x_p^{n-1} \gets x_p^{n-2}$ \label{se7} \Comment{Use past delays}
      \Else\label{se8}
        \State $x_p^{n-1} \gets d_{\max}+\epsilon$ \label{se9}\Comment{Speculate the delay}
      \EndIf
    \EndIf
  \EndFor
  \For {$p \in \mathcal P\setminus \mathcal P_e^{n-1}$}\label{se10} \Comment{For EPs not used for $\text{TX}^{n-1}$}   
    \State $x_p^{n-1} \gets x_p^{n-2}$ \label{se11} \Comment{Use past delays}
  \EndFor
  \State $\mathcal P_e^{n} \gets$ {\tt Eligibile-EPs-with-min-delay}($|\mathcal P|/2$,$X^{n-1}$)\label{se12}
  \State $p\gets$ {\tt Random-EP}($\mathcal P\setminus (\mathcal P_e^n \cup  \mathcal P_e^{n-1})$) \label{se13} \Comment{Select probe EP}
  \State $\mathcal P_e^{n} \gets$ {\tt Replace-slowest-EP}($\mathcal P_e^n, p$)\label{se14}\Comment{Embed the probe EP}
  \State \Return $\mathcal P_e^n$ \label{se15} \Comment{Selected EPs augmented with the probe EP}
\EndProcedure
\end{algorithmic}
\caption{Pseudocode for {\tt SelectEndorsers} }\label{alg:se}
\end{figure}

%% file: Sec5_NumEvaluation.tex
\section{Performance evaluation}\label{sec:simu}

\subsection{Methodology}\label{sec:met}
We developed an event-driven simulator using OMNeT++~\cite{omnet}. 
We considered a scenario with $C=8$ clients and $Q=8$ organizations, each of them with $1$~EP, thus $|\mathcal P|=Q$. The endorsement requests are generated according to a Poisson process at each client and we set the normalized load $\gamma \in [0.1,0.9]$. Then fixing $\gamma = 0.5$, we considered more scenarios by varying $Q \in \{8,16,32,64\}$, each organization with the number of $\text{EP}\in\{1,2,4,8\}$, and $C \in \{8,40,125,1000,8000,32000\}$ clients; in each scenario we fixed all parameters except one.
To understand the performance under non-stationary requests, we also considered a Poisson-modulated process with squared-wave cyclo-stationary load, with a period equal to 1200~ms, duty cycle 50\%, and normalized load $\gamma = 0.5$. 
To consider the effect of different kinds of computation, we assume the computation time of each EP to be either exponentially distributed\tagged{trp}{, bi-modal distributed, mixed bi-modal with exponential distributed or log-normally distributed}\tagged{jrn}{ or bi-modal distributed} with an average equal to $10$~ms, whose value has been achieved from our practical measurements in HF EPs. In the bi-modal case, we assumed that, with a given probability, the computation time is 
constant with the value $1/\mu_1$, otherwise its value is $1/\mu_2$.
\tagged{trp}{In the mixed bi-modal with the exponential case, we assumed that, with a given probability, the computation time is exponentially distributed and has an average of $1/\mu_1$, otherwise its average is $1/\mu_2$. In the log-normal case, we assumed that, with a given probability, the computation time is In the log-normally distributed scenario, the average computation time is equal to $10$~ms, while the standard deviation is changing.}
Table~\ref{table:cvEXP} 
shows the coefficient of variations (Cv) for 
\tagged{trp}{the adopted setting for the log-normal and both of the bi-modal cases.}
\tagged{jrn}{the adopted setting.}
To model the heterogeneity in the computing power and resources of the EPs, we considered a {\em non-homogenous scenario} in which we assigned different average computation times to different EPs (i.e., ($2,4,6,8,12,14,16,18$)~ms) where the computation time of each EP is exponentially distributed. 
\tagged{jrn}{
\begin{table}[!tb]
    \centering
    \caption{Settings for different distributions of the computation time with different coefficient of variation (Cv).}
    \begin{tabular}{|c||c|c|c|c|c|}
        \hline
         Cv (Bimodal) & $0.0$ & $0.5$ & $1$ & $2$ & $5$ \\ 
        \hline\hline
         $P(\mu=\mu_1)$  & $0.5$ & $0.6$ & $0.75$ & $0.9$ & $0.98$ \\ 
         \hline
         {$1/\mu_1$~[ms]} & $10.0$ & $5.9$ & $4.2$ & $3.4$ & $2.85$ \\ 
         {$1/\mu_2$~[ms]} & $10.0$ & $16.1$ & $27.3$ & $70.0$ & $360.0$ \\
        \hline
    \end{tabular}
    \label{table:cvEXP}
\end{table}}
\tagged{trp}{
\begin{table}[!tb]
    \centering
    \caption{Settings for different distributions of the computation time with different coefficient of variation (Cv).}
    \begin{tabular}{|c||c|c|c|c|c|}
        \hline
         Cv (Bimodal) & $0.0$ & $0.5$ & $1$ & $2$ & $5$ \\ 
        \hline
         $P(\mu=\mu_1)$  & $0.5$ & $0.6$ & $0.75$ & $0.9$ & $0.98$ \\ 
         {$1/\mu_1$~[ms]} & $10.0$ & $5.9$ & $4.2$ & $3.4$ & $2.85$ \\ 
         {$1/\mu_2$~[ms]} & $10.0$ & $16.1$ & $27.3$ & $70.0$ & $360.0$ \\
        \hline
        \hline
         Cv (log-normal) & $0.001$ & $0.5$ & $1$ & $2$ & $5$ \\ 
        \hline
         {$\mu_Z$} & $2.3025$ & $2.191$ & $1.956$ & $1.498$ & $0.673$ \\ 
         {$\sigma_Z$} & $0.0010$ & $0.472$ & $0.832$ & $1.268$ & $1.805$ \\
        \hline
    \end{tabular}
    \label{table:cvEXP}
\end{table}}

We considered three scenarios for the network model, two of them are synthetic and the last one is real. Let $d_{cp}$ be the network delay between client $c$ and EP $p$. In the first scenario, denoted as S1, the network delays are negligible compared to the processing times at the EP, i.e., $d_{cp}=0$~(Fig.~\ref{fig:S1}). In the second scenario, denoted as S2, we set linearly increasing delays between any client and the EPs, similarly to a linear topology where all clients are closer to the first EP, i.e., $d_{cp}=(p+1/2)$~ms for $p\in [1, Q]$. This implies similar delays from each EP to any client while on average the total network delays are comparable to the processing times at the EPs~(Fig.~\ref{fig:S2}). 
\begin{figure}[!tb]
\minipage{0.5\columnwidth}
  \centering
  \includegraphics[width=0.7\textwidth]{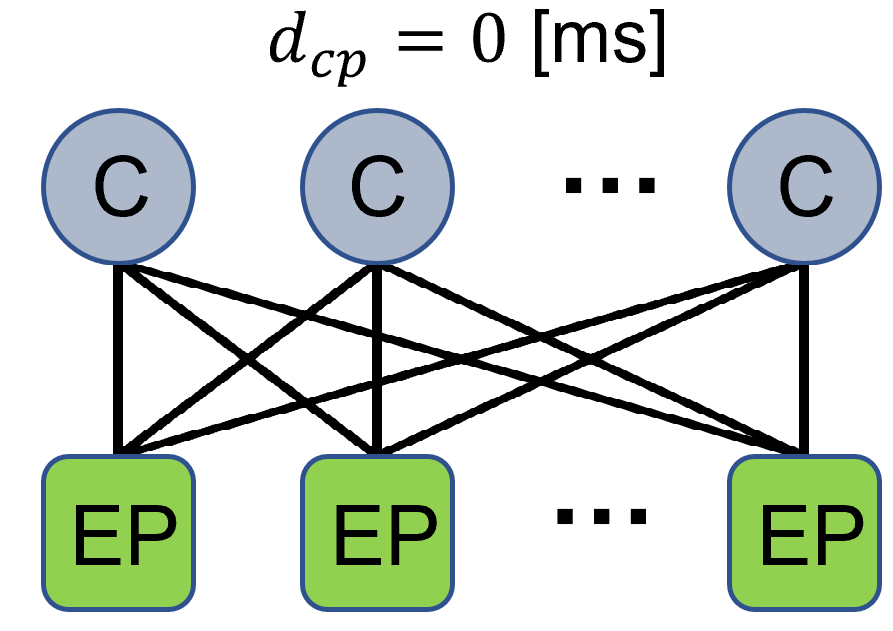}
  \subcaption{S1 (bipartite)}
  \label{fig:S1}
\endminipage
\minipage{0.5\columnwidth}
  \centering
  \includegraphics[width=0.7\textwidth]{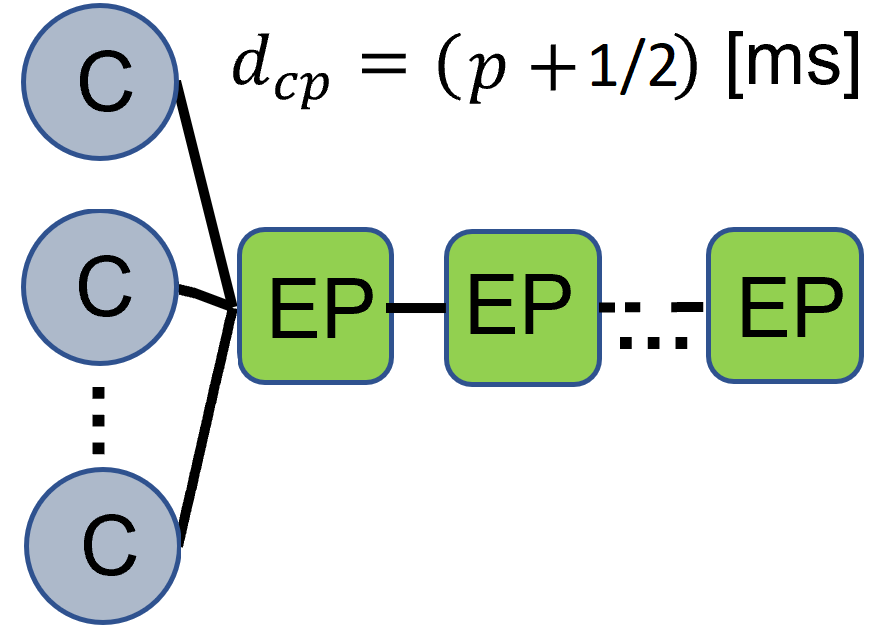}
  \subcaption{S2 (linear)}
  \label{fig:S2}
\endminipage
\caption[width=1\columnwidth]{{The two synthetic network topologies adopted for test scenarios in our simulations.}}
\label{figure:topos}
\end{figure}

In the third scenario, denoted as S3, we selected the \textit{Highwinds} network from~\cite{intTopoZoo}, shown in Fig.~\ref{fig:highwinds}, as a real world-wide scenario where the link delays are calculated based on the physical distance between the geographical position of the nodes (using the {Haversine formula}) and the propagation speed is $2/3$ the speed of light. The clients here can be divided into two groups: (i) {\em far clients} placed in nodes $1,7,8$, and (ii) {\em centered clients} placed in nodes $2,3,4,5,6$.
\begin{figure}[!tb]
  \centering
  \includegraphics[width=8cm]{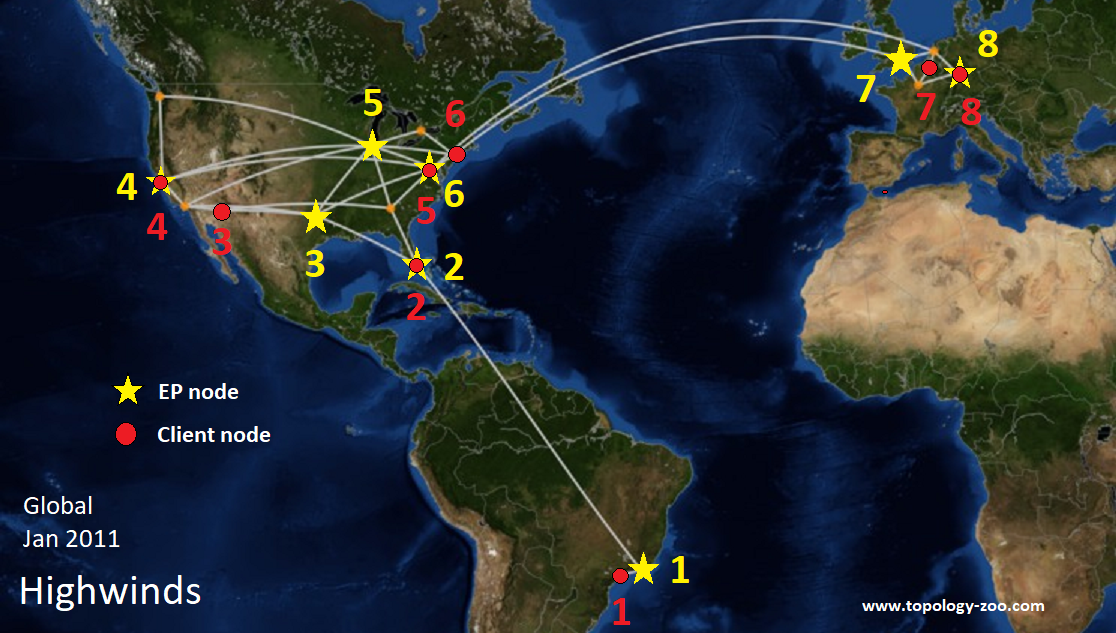}
  \caption{{Real network topology (S3) showing the EPs and clients placement with interconnecting network topology}}
  \label{fig:highwinds}
\end{figure}


We measure the average \en latency as the main performance metric. The \en latency is calculated from the moment the \en request is sent out from the client until the first response is received by the client. 
For comparison, we considered three EP selection algorithms, namely {RND}, {OOD}, and {DSLM}, where the first two are proposed by us. 

\subsubsection{Random EPs (RND)}
{RND} is the policy adopted in the analytical model of Sec.~\ref{sec:mod}. Every \en request is sent to $R$ randomly chosen EPs. 
If $R=Q/2$,  the policy is denoted as {\em RND-half}. If $R$ adapts to the load according to the rule $R=Q/(2\gamma)$, as suggested by~\eqref{eq:1rhat3}, we denote the policy as {\em RND-load}.

\subsubsection{Dynamic Stochastic Load Minimization (DSLM)}
Dynamic Stochastic Load Minimization (DSLM) was proposed in~\cite{soa} and the pseudocode of the version adapted to our system model is  
shown in Fig.~\ref{alg:dslm2}. 
Just for the first TX, {DSLM} initializes the load $l_p$ and the measured response delay $x_p^0$ of any EP $p$ (ln.~\ref{d2}-\ref{d4}). Typically, it randomly selects half of the EPs~(ln.~\ref{d5}) and evaluates heuristically the load on each selected EP by the product of the square root of the response delay and the corresponding queue length~(ln.~\ref{d6}-\ref{d7}). The average is obtained with an exponential moving average with parameter $\alpha$. Finally, {DSML} returns the EP with the lowest estimated load among the selected ones~(ln.~\ref{d20}).
\begin{figure}[!tb]
\scriptsize
\begin{algorithmic}[1]
\Procedure{DSLM}{$n$}\Comment{Process TX$^n$}
  \If{$n=1$} \label{d2} \Comment{Just for the first TX}
      \State $l_p \gets \bar{x}_p \gets x_p^0 \gets 0, \forall p \in \mathcal P$ \label{d4} \Comment{Init EP load and delay values}
  \EndIf
  \State $\mathcal P_{h} \gets \texttt{Select-}|\mathcal P|/2$\texttt{-random-peers}\label{d5}\Comment{Random half EPs}
  \For {$p \in \mathcal P_{h}$}\label{d6}
     \State $\bar{x}_p \gets [ \alpha \bar{x}_p+(1-\alpha) x_p^{n-1}]$
    \label{d7} \Comment{Average delay}
  \EndFor
\State \Return $\arg\min_{p \in \mathcal P_h} \{({\bar{x}_p^{0.5}}+1)q_p\}$\label{d20}\Comment{Choose min product delay queue length.} 
\EndProcedure
\end{algorithmic}
\caption{Pseudocode for {DSLM} adapted to our model}\label{alg:dslm2}
\end{figure}

\subsubsection{Oracle Optimal Delays (OOD)}

As a reference for all the \en algorithms, we define an online Oracle-based Optimal Delays (OOD) EP selection policy that minimizes the \en latency given a fixed replication factor $R$, 
denoted as OOD-$R$. The pseudocode of {OOD-$R$} is provided in Fig.~\ref{alg:ood}. We assume an oracle that knows in advance the response delay of any \en request if sent to a specific EP. Thus, the oracle knows for any EP $p$: (i) the absolute time $t^\text{busy}_p$ at which the EP will finish (or has finished) to serve the last received \en request TX$^{n-1}$, (ii) the processing time $\tau_p^\text{proc}$ of the \en request TX$^{n}$, and (iii) the overall network delay $\tau^\text{net}_p$ between each client and the EP. 
{Thus, if sent to EP $p$, the response to TX$^n$ will be received from EP $p$ at a predicted time $t^\text{resp}_p$~(ln.~\ref{o4p}) equal to:
\begin{equation}\label{eq:oracle}
    t^\text{resp}_p= \max\{t^\text{now}+\tau^\text{net}_p,t^\text{busy}_p\}+\tau_p^\text{proc}+\tau^\text{net}_p
\end{equation}
since if at the time of arriving the request to an EP its queue is empty, then the request will be served at $t^\text{now}+\tau^\text{net}_p$, otherwise at $t^\text{busy}_p$. Then, the request will be processed for $\tau_p^\text{proc}$ and the response will be sent back, experiencing $\tau^\text{net}_p$ delay. 
Now OOD-$R$ chooses the EP with the smallest predicted time to minimize the response delay~(ln.~\ref{o5}).
The remaining $(R-1)$ \en requests (if any) will be sent to the EPs in decreasing order of predicted time (ln.~\ref{o7p}). This allows to load the ``slowest'' EPs with requests whose responses will be received late and thus reduces the load on the ``fastest" EPs, for the sake of future \en requests. 

It should be noted that, in the case of RND-load, the request arrival is assumed to be stationary, thus, the system load can be estimated with high accuracy.
Also, OOD is implementable with enough control information, but obtaining this information would need instantaneous communication with the EPs, which is challenging to accomplish in a practical situation.
 So, both algorithms are not practical in a real scenario.
\begin{figure}[!tb]
 \scriptsize
 \begin{algorithmic}[1]
 \Procedure{OOD-$R$}{$n$} \label{o1p}\Comment{Process TX$^n$}
 \For {$p \in \mathcal P$}\label{o2p}\Comment{For each EP}
    \State $ t^\text{resp}_p= \max\{t^\text{now}+\tau^\text{net}_p,t^\text{busy}_p\}+\tau_p^\text{proc}$
   \label{o4p}\Comment{Compute response delay}
 \EndFor
 \State $p=\arg \min_{p\in\mathcal P} \{ t^\text{resp}_p\} $ \label{o5}\Comment{Choose the min response delay EP}
 \State $\mathcal P_e=
 \texttt{Find-}(R-1)\texttt{-EPs-with-largest-}\{t_p^\text{resp}\}$\label{o7p}
 \State \Return $\{p\} \cup \mathcal P_e $\label{o18}\Comment{Return endorsement set}
\EndProcedure
\end{algorithmic}
\caption{Pseudocode of OOD-$R$  }\label{alg:ood}
\end{figure}
}




\subsection{Simulations results}
For a fair comparison between {OOD} and other approaches, in all test scenarios 
we selected OOD-half, i.e., with the same $R$ as OPEN, and RND-half, and slightly smaller $R$ than RND-load, 
for which $R\in [Q/2,Q]$. 
Only DSLM has a completely different redundancy factor ($R=1$).



\subsubsection{Homogenous scenario}

The left graphs in Fig.~\ref{figure:avgDelay} show the simulation results for a homogenous scenario with all EPs with computation times that are exponentially distributed with the same average.

\begin{figure*}[!tb]
  \includegraphics[width=\linewidth,height=10cm]{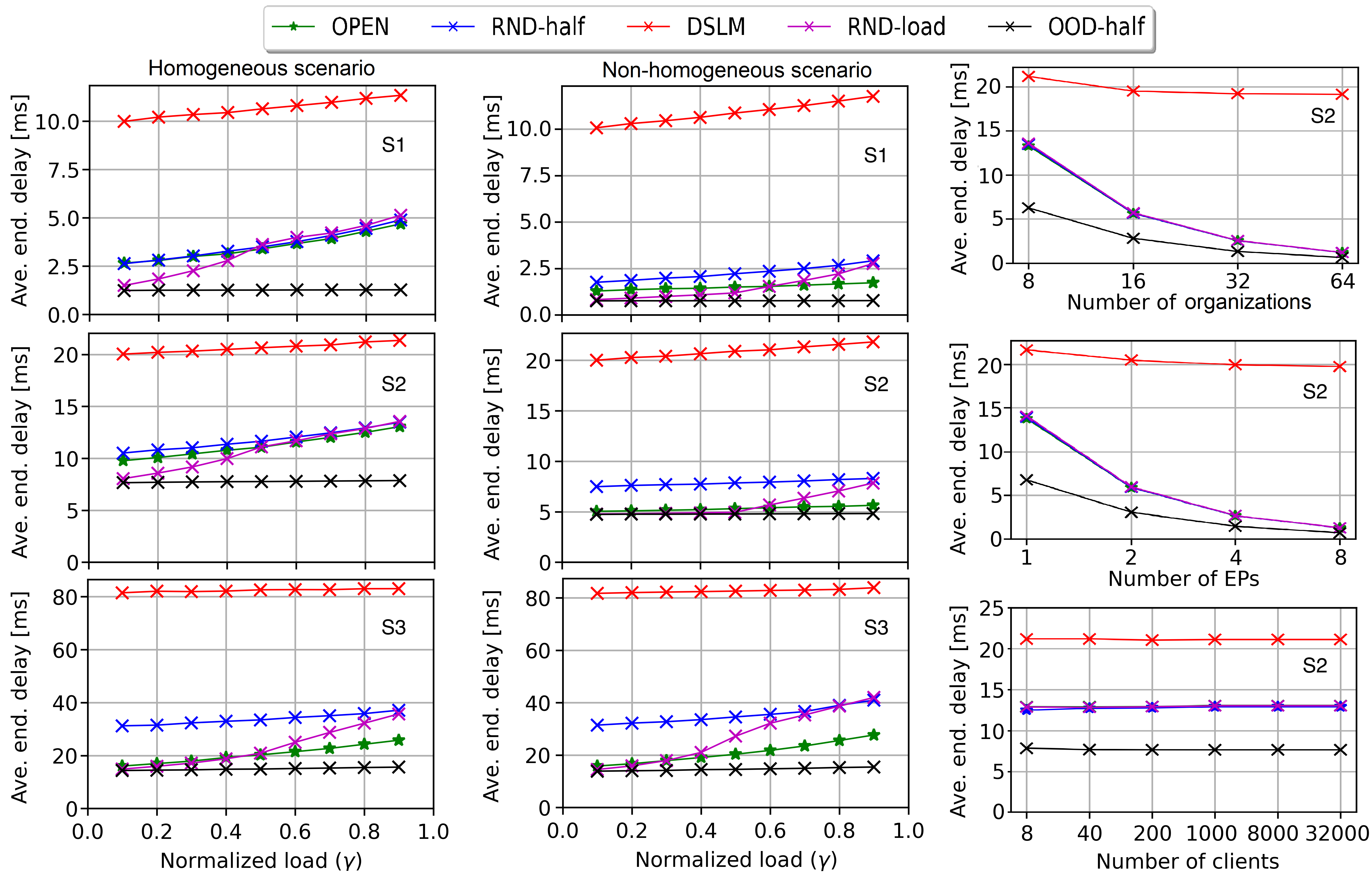}
  \caption{Average endorsement delay for i) average computation time of~10~ms for each EP: S1 (left-up), S2 (left-middle), S3 (left-down), ii) different average computation times from~[2 to 18]~ms for each EP: S1 (middle-up), S2 (middle-middle), S3 (middle-down), and iii) the different number of organizations, EPs, and clients in scenario S2, for $\gamma=0.5$ (right-most graphs).}
  \label{figure:avgDelay} 
\end{figure*}

In scenario S1 (left-up), all delays are purely due to processing in the EPs. Since {DSLM} does not exploit redundancy and for $\gamma=0.1$ the queuing at the EP is negligible, its response delay is around 10~ms, equal to the computation time. By increasing the load and hence the queueing, the average delay increases slightly. Instead, by exploiting the redundancy all the other approaches can get smaller delays, by a factor of $2 \text{ to } 6$. 
As expected, {OOD-half} achieves the best average \en latency among all solutions. At low loads, due to the maximum redundancy factor (i.e., $R=8$), {RND-load} performs closer to {OOD-half} by always having the fastest EP among its selection. By increasing the arrival rate, for both {RND-half} and {RND-load}, the \en latency increases as the selection of EPs is not efficient as {OOD-half}, which knows in advance the best EP.
At high loads, for {RND-load}, $R$ is almost 4, hence {RND-load} shows similar results to {RND-half}. 
 {OPEN} has a redundancy factor $R=4$, as {RND-half}, but selects EPs with smaller estimated delays. At low loads, {OPEN} has a small advantage over {RND-half}, as the queueing is almost negligible.
As the offered load increases, the higher queueing makes OPEN more efficient, also thanks to the higher frequency by which the response delays are estimated.

In scenario S2 (left-middle), as expected, {OOD-half} is the best algorithm, and {DSLM} is outperformed by all other solutions by a factor of $2 \text{ to } 3$. Due to the linearly increasing network delays, the effect of redundancy in EPs becomes less dominant, so the delay's improvement in scenario S2 is less than in S1. 
On the other hand, at low loads, {OPEN} acts slightly better than {RND-half} compared to S1, thanks to being aware of the network delays. At high loads, OPEN behaves close to RND-half since the queueing delays become dominant to network delays.

In scenario S3 (left-down), again {OOD-half} is the best approach, {DSLM} is outperformed by all other solutions by at least a factor of $2$. Due to the different network delays, on average much larger than the computation times, the redundancy is less effective, thus a lower delay improvement is experienced in S3 compared to S2, and S1. {OPEN} performs quite similarly to RND-load in low loads even with a half number of selected EPs, and much better in high loads. {OPEN} completely outperforms RND-half in all loads since it exploits mainly the EP with lower network delays.

\subsubsection{Non-homogenous scenario}
The simulation results for a non-homogenous scenario are reported in the middle graphs of Fig.~\ref{figure:avgDelay}. 
As a reminder, now the average computation times for the EPs are different, but the overall average is the same as in the homogenous scenario.
In all three scenarios S1, S2, and S3, as DSLM is not able to exploit redundancy, it is not able to reduce its average latency. On the other hand, by exploiting redundancy, all other approaches can reduce their average latency, where OOD-half achieves the lowest latency thanks to its global knowledge of the system.

In scenario S1, RND-load reduces the latency more than RND-half, since the higher redundancy factor increases the chance of selecting EPs with lower average computation times. With the same redundancy factor as RND-half, OPEN reduces the most the delays for all loads, as it employs the latency history to select EPs with lower average computation times.
In scenario S2, a similar behavior as in S1 is observed for all the algorithms. OPEN, by exploiting the delay history comprising both the average computation times and the network delays, achieves the best performance by almost approaching OOD-half. 
In scenario S3, we observe almost similar results as in scenario S3 of the homogeneous case, as the variation in the average computation times is still negligible to the average network delays.

\subsubsection{Scaling the number of organizations, EPs, and clients.}
The simulation results for larger scenarios are shown in Fig.~\ref{figure:avgDelay} (right). 
We consider the S2 scenario, to get a heterogeneous system in terms of network delays, and we fixed $\gamma = 0.5$.

By increasing the number of organizations, the overall number of EPs increases, thus the \en latency is reduced for all the algorithms exploiting redundancy, as shown in Fig.~\ref{figure:avgDelay} (right-up). The same behavior is observed when the number of EPs in each organization increases (see Fig.~\ref{figure:avgDelay} (right-middle)). The similarity with the previous graph is that we are considering the $1$-OutOf-$N$ policy here, which by recalling~\eqref{pr:p2}, for this \en policy there is no difference between two EPs of the same organization or different EPs of different organizations. 

According to Fig.~\ref{figure:avgDelay} (right-down), changing the number of clients has no effect on the approaches. Note that increasing the number of clients will reduce the efficiency of the information gained by OPEN and it will converge to the RND-half results for homogeneous cases with less dominant network delays.

\subsubsection{Bi-modal computation times}

\tagged{trp}{The simulation results for constant bi-modal are shown in Fig.~\ref{figure:Cv1d0} and Fig.~\ref{figure:Cv1d1}.}
\tagged{jrn}{The simulation results are shown in Fig.~\ref{figure:Cv1d}.}
In scenario S1, for constant computation time (Cv=$0$), the redundancy is not beneficial for delay reduction, while at high loads it can increase the EPs’ queue length and thus the delay. For larger Cv, all the algorithms, except DSLM, decrease the average endorsement delay. This is because the average of the minimum between a sequence of i.i.d. random variables is smaller when the variance is larger. All the solutions, except for DSLM, behave similarly for low and high loads.

In scenario S2, also for Cv=$0$, the redundancy reduces the average delay. The reason is that DSLM considers the computation load at the EPs obliviously of the network delays, which are dominating the computation times. But, in the other approaches, redundancy increases the chance of selecting the EP with lower network delays. By increasing Cv, redundancy can reduce the latency even more, by benefiting from the variability in the computation times. At low load ($\gamma = 0.2$), OPEN performs quite well as it also selects EPs with lower network delays. RND-load is performing slightly better as it sends to all EPs. OOD-half is even better than RND-load with a small margin, thanks to the lower load guaranteed by setting $R = 4$. At high load ($\gamma = 0.8$), RND-load adopts $R = 5.7$ (on average) and the corresponding queueing penalizes the overall response delay. OPEN acts slightly better thanks to the smaller value of $R$.

In the S3 scenario, all approaches are not affected by Cv, as the variability in the computation times is compensated by the network delays which vary between $0$~ms and $7$~times the average computation time. At low load ($\gamma = 0.2$), OPEN selects closer EPs in terms of network delays and outperforms RND-half by a factor greater than $2$, while being very close to OOD-half. RND-load achieves the same results as OPEN by selecting all the EPs (i.e., $R = 8$), which include the closest EP as well. At high load ($\gamma = 0.8$), as in scenario S2, RND-load is penalized by the queueing. OPEN reduces the endorsement delay up to $70\%$ compared to DSLM. Notably, differently from OPEN, RND-load may not select the closest EPs. As expected, for both loads OOD-half performs the best, since it always selects the minimum combination of the network delay and the processing delay.
\tagged{trp}{
\begin{figure*}[!tb] 
    \includegraphics[width=\linewidth,height=4cm]{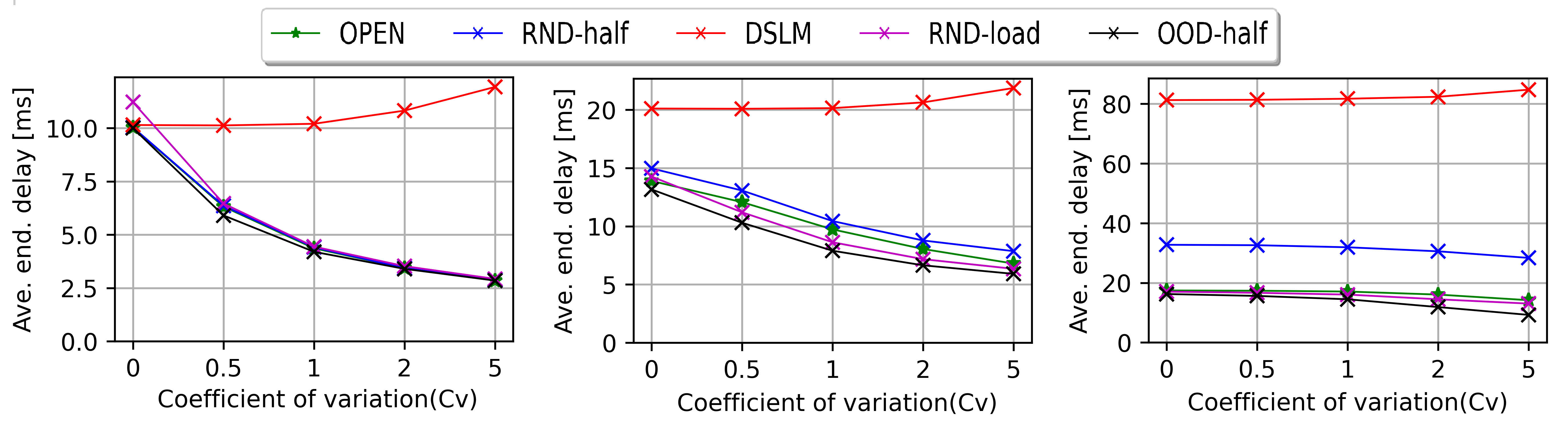} 
  \caption{Average \en latency under bimodal computation times for normalized load $\gamma=0.2$ and for different scenarios: S1 (left), S2 (middle), S3 (right)
  .}
  \label{figure:Cv1d0}
\end{figure*}
\begin{figure*}[!tb] 
    \includegraphics[width=\linewidth,height=4cm]{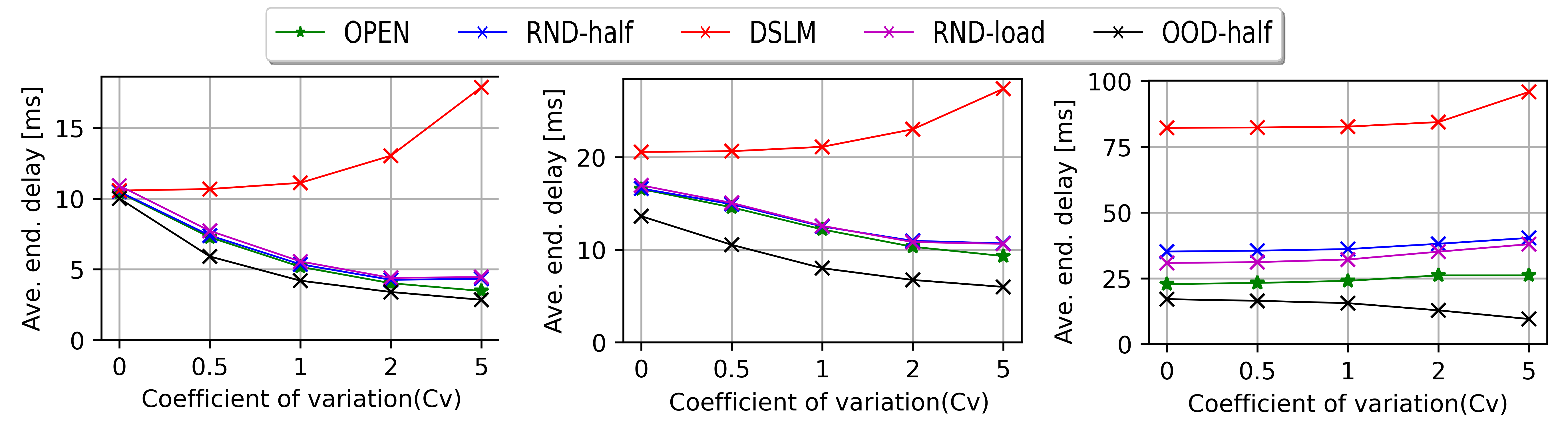} 
  \caption{Average \en latency under bimodal computation times for normalized load $\gamma=0.8$ and for different scenarios: S1 (left), S2 (middle), S3 (right)
  .}
  \label{figure:Cv1d1}
\end{figure*}
}
\tagged{jrn}{
\begin{figure}[!tb] 
    \includegraphics[width=\columnwidth,height=8cm]{Figures/S012_EP1_C1_Q8servEXP(10)Bimodal(0-5-10-20-50)all-vertical (1).png} 
  \caption{Average \en latency under bimodal computation times: i) normalized load: $\gamma=0.2$ (left), $\gamma=0.8$ (right), ii) scenarios: S1 (up), S2 (middle), S3 (down)}
  \label{figure:Cv1d}
\end{figure}
}
%
%
\tagged{trp}{
\subsection{Mixed bi-modal with exponential computation times}
The simulation results for constant bi-modal are shown in Fig.~\ref{figure:Cv1d2} and Fig.~\ref{figure:Cv1d3}.
\begin{figure*}[!tb] 
    \includegraphics[width=\linewidth,height=4cm]{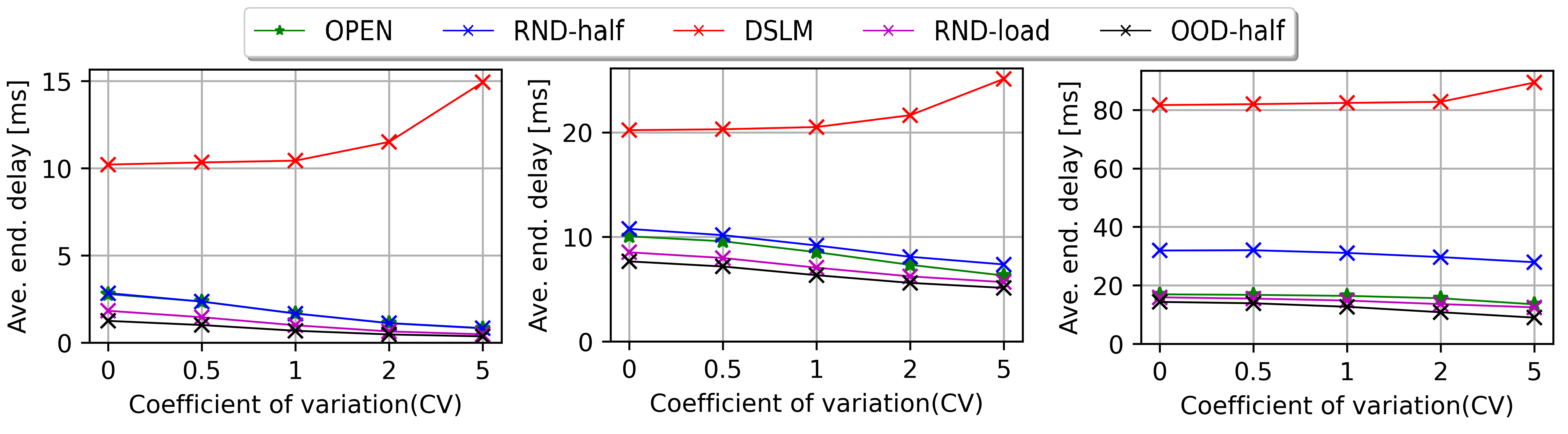} 
  \caption{Average \en latency under mix bimodal with exponential computation times for normalized load $\gamma=0.2$ and for different scenarios: S1 (left), S2 (middle), S3 (right)
  .}
  \label{figure:Cv1d2}
\end{figure*}
\begin{figure*}[!tb] 
    \includegraphics[width=\linewidth,height=4cm]{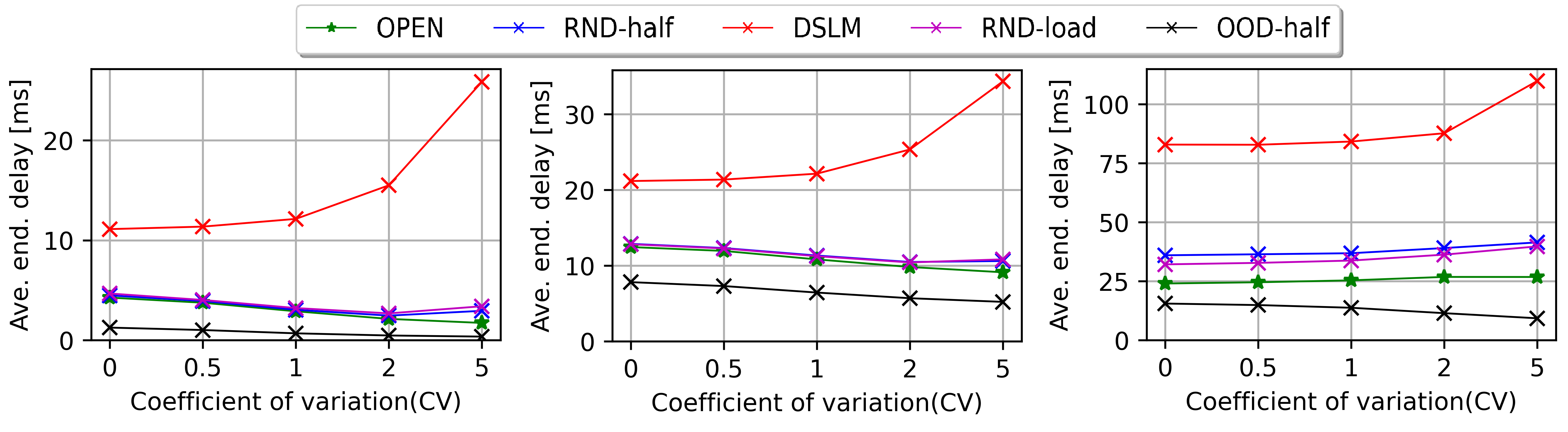} 
  \caption{Average \en latency under mix bimodal with exponential computation times for normalized load $\gamma=0.8$ and for different scenarios: S1 (left), S2 (middle), S3 (right)
  .}
  \label{figure:Cv1d3}
\end{figure*}

For all scenarios, when Cv=$0$, all EPs have one computation time which is exponentially distributed with an average of $10$~ms. So, the experienced latency in low loads and high loads is similar to what is reported in Fig.~\ref{figure:avgDelay} (left) when $\gamma=0.2$ and $\gamma=0.8$ respectively. For all scenarios, the redundancy reduces the average delay. 

In scenario S1, by increasing the CV, all the algorithms, except {DSLM}, decrease the average endorsement delay. This is because the average of the minimum between a sequence of i.i.d.\ random variables is smaller when the variance is larger. 
At low load ($\gamma=0.2$), {RND-load} is performing a bit better than OPEN and RND-half, as it sends to all EPs. {OOD-half} is slightly better than RND-load thanks to the lower load guaranteed by setting $R=Q/2$.
At high load ($\gamma=0.8$), {RND-load} adopts $R>Q/2$ and the corresponding queueing penalizes the overall response delay.
By increasing the CV, for DSLM in both low loads and high loads, the average endorsement delay increases since even a small possibility of selecting an EP with very high computation time can increase the average \en delay.

In scenario S2, similar to scenario S1, in both low loads and high loads by increasing the CV, the average endorsement delay for DSLM increases. 
In the other approaches, redundancy increases the chance of selecting the EP with lower network delays. By increasing Cv, redundancy can benefit from more variability in the computation times, and reduce the latency even more.
At low load ($\gamma=0.2$), {OPEN} performs slightly better than RND-half as it also selects EPs with overall lower network delays and processing delays. {RND-load} is performing even better as it sends to all EPs. 
{OOD-half} is even better with a small margin, thanks to the lower load guaranteed by setting $R=Q/2$.
At high load ($\gamma=0.8$), {RND-load} adopts $R>Q/2$ and the corresponding queueing penalizes the overall response delay. {OPEN} acts slightly better thanks to the smaller value of $R$.

In the S3 scenario, all algorithms benefiting from the redundancy, are not affected by Cv, as the variability in the computation times is compensated (i.e., equalized) by the network delays which vary between $0$~ms and $10$ times the average computation time. 
At low load ($\gamma=0.2$), {OPEN} being very close to {OOD-half}, outperforms {RND-half} by a factor around $2$ due to selecting closer EPs in terms of network delays. {RND-load} achieves the same results as {OPEN} by containing the closest EP as it selects all the EPs (i.e., $R=Q$). 
At high load ($\gamma=0.8$), same as scenario S2, OPEN reduces the \en latency up to 70\% compared to DSLM. {RND-load} is penalized by both the queueing and lower number of selected EPs (i.e., $R<Q$), as differently from OPEN, RND-load may not select the closest EPs.
Again, for both loads {OOD-half} performs the best, since it always selects the minimum combination of the network delay, the queueing delay, and the processing delay. 
 

\subsubsection{Log-normal computation times}

The simulation results are shown in Fig.~\ref{figure:Cv1d4}.
All the results are similar to the results of high loads ($\gamma=0.8$) in Bi-modal computation times, except for scenario S1, in which for constant computation time (Cv=$0$), redundancy is not beneficial for delay reduction. For larger Cv, all the algorithms, except {DSLM}, decrease the average endorsement delay as the average of the minimum between a sequence of i.i.d.\ random variables is smaller when the variance is larger.
\begin{figure*}[!tb] 
    \includegraphics[width=\linewidth,height=4cm]{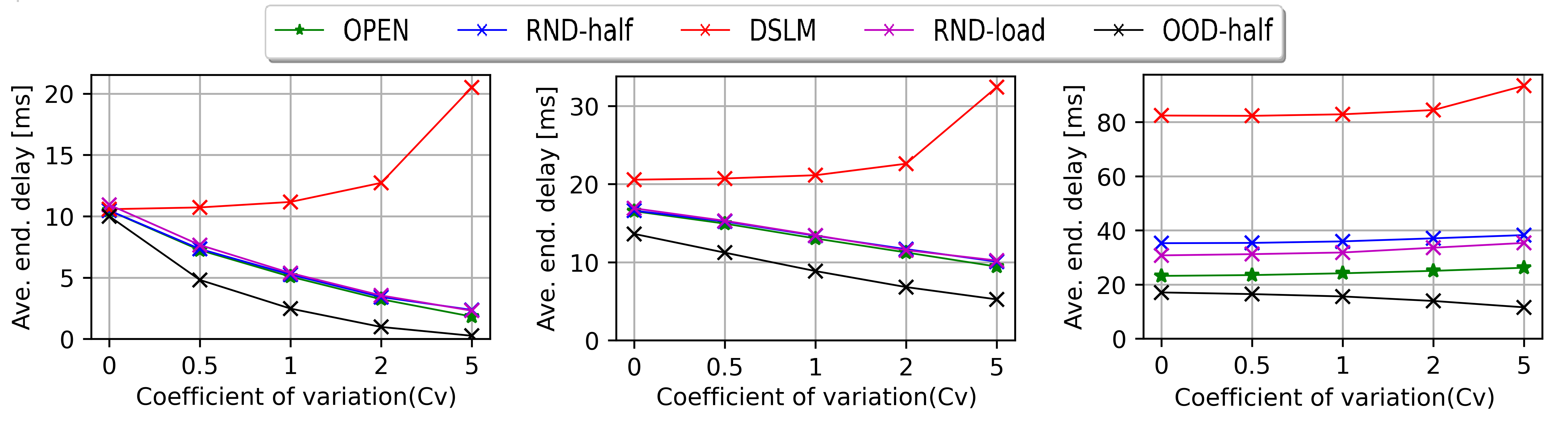} 
  \caption{Average \en latency under log-normal computation times for normalized load $\gamma=0.8$ and different scenarios: S1 (left), S2 (center), S3 (right).}
  \label{figure:Cv1d4}
\end{figure*}
}

\subsubsection{Cyclo-stationary request process}
We compared {OPEN} with other approaches under Poisson-modulated cycle-stationary load. We evaluated the average endorsement delays by using an exponential moving average. The results are provided in Table~\ref{table:nonsta1}.
\begin{table}[!tb]
    \centering
    \caption{Average endorsement delays for cyclo-stationary input rates with different scenarios.}
    \begin{tabular}{@{}|c||cccc|@{}}
        \hline
         Scenario & S1~[ms] & S2~[ms] & S3~[ms] & S3~[ms] \\
         & & & centered clients & far clients \\
         \hline    
         OOD-half & $1.2$ & $7.1$ & $13.6$ & $16.9$ \\
         OPEN & $2.9$ & $11.7$ & $17.9$ & $24.2$ \\
         RND-load & $3.1$ & $11.2$ & $19.2$ & $23.7$ \\
         RND-half & $3.2$ & $12.2$ & $25.9$ & $43.6$ \\
         DSLM & $10.8$ & $22.8$ & $65.9$ & $110.3$ \\    
         \hline
    \end{tabular}
    \label{table:nonsta1}
\end{table}
In S1 and S2, OPEN, RND-half, and RND-load showed almost constant average \en latency, while DSLM and ODD results are the highest and the lowest respectively. Interestingly, all the results for different approaches in S1 and S2 are very close to the results gained from Fig.~\ref{figure:avgDelay} (left) for $\gamma=0.5$, even if the load was changing periodically. This means that all of them are robust to load change in homogeneous scenarios.   

In S3, as a non-homogenous real scenario, {OPEN} shows a small difference of the average \en latency between 
centered and far clients (recall their definition in Sec.~\ref{sec:met}). This difference ($6$~ms) is negligible compared to the average network delays in S3 ($50$~ms).
The same behavior is observed for RND-load. 
On the other hand, in {RND-half} the performance depends heavily on the client's position; even in the case of centered clients, RND-half experiences more \en latency than OPEN with far clients.
As expected, OOD-half achieves the best \en latency with minimum difference regardless of the client's position. DSLM performs the worst with latencies about 4 times larger than OPEN.

These results show that {OPEN} adapts to load changes even in the presence of unbalanced network delays. Also, OPEN outperforms RND-half and DSLM, while it shows similar results to RND-load and to OOD-half.   

\tagged{trp}{
\section{Proof-of-concept implementation}\label{sec:poc}
We have validated OPEN in a practical scenario, by implementing a proof-of-concept solution in a real HF platform and testing it. Our preliminary results show that OPEN is easily implementable in HF, just with some coding on the client with no modification to HF core-engine code. 
Using HF v2.2, $8$ EPs are implemented as Docker containers on a single server and connected directly through a Docker virtual network, following the setup of  Fig.~\ref{fig:system}.
A background load of endorsement requests is created across all the EPs such that each EP is loaded differently from the others, according to Table~\ref{tab:back}. Fig.~\ref{fig:openrnd4} shows the corresponding processing delays at the different EPs, which are unbalanced as expected. In our settings, the network delays are practically negligible. 

\begin{table}[!tb]
\begin{center}
\caption{Background load for the EPs in the testbed 
}\label{tab:back}
\begin{tabular}{|c|c|c|}
\hline
Endorser Peer id & Request rate [1/s]\\
\hline
1 & 1.2\\
2 & 2.3\\
3 & 3.5\\
4 & 4.7\\
5 & 5.9\\
6 & 7.1\\
7 & 8.3\\
8 & 9.5\\
\hline
\end{tabular}
\end{center}
\end{table}

\begin{figure}[!tb]
	\begin{center}
        \includegraphics[width=0.9\linewidth]{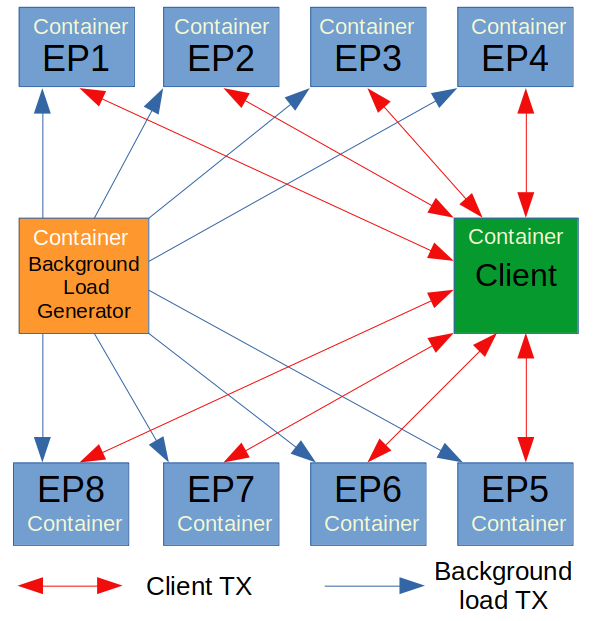}
        \caption{Experimental testbed architecture 
        }
        \label{fig:system} 
        \end{center}
\end{figure}        
 The measured endorsement delay is reported in Table~\ref{tab:1} shows that OPEN can reduce the endorsement delay by $22$\% compared to RND-half. We can deduce that, in this non-homogeneous case, OPEN can ``learn" the best EPs where to send the request.
\begin{figure}[!tb]
\begin{center}
        \includegraphics[width=0.9\linewidth]{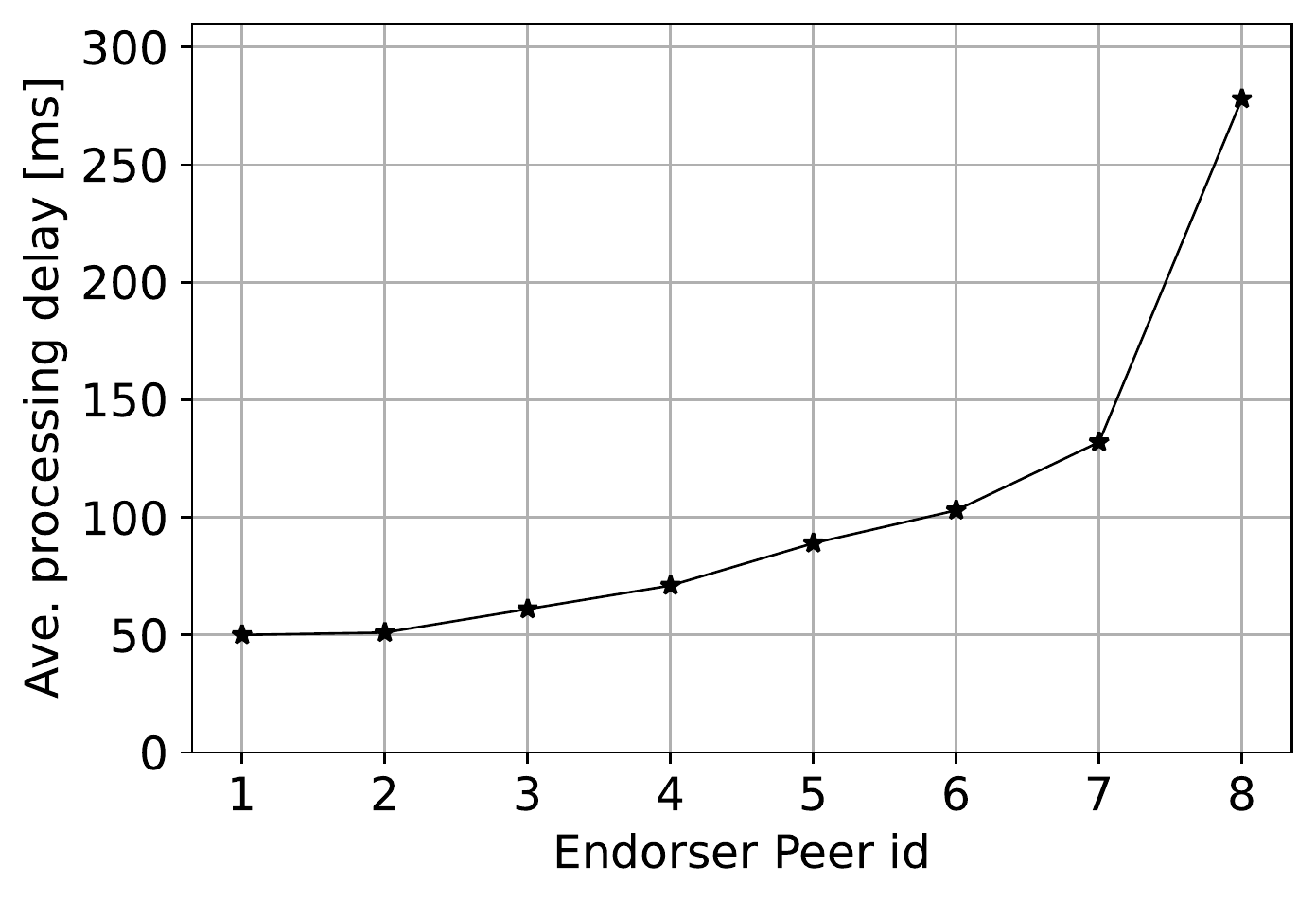}
                \end{center}
\caption{Measured processing delay at each EP.}
\label{fig:openrnd4} 
\end{figure}

\begin{table}[!tb]
\begin{center}
\caption{Experimental results in the considered testbed}\label{tab:1}
\begin{tabular}{|c|c|c|}
\hline
Algorithm & Ave.\ endorsement delay & $95$\% confidence interval \\
\hline
RND-half & $115$~ms & $[110,120]$~ms \\ 
OPEN & $94$~ms & $[87,101]$~ms\\ 
\hline
\end{tabular}
\end{center}
\end{table}

}

%% file: Sec6_RelWorks.tex
\section{Related works}\label{sec:related}


Different works modeled analytically the \en process in HF. \cite{end3} modeled the EPs as M/M/1 queues and considered the propagation delays in the network model, coherently with our work. It showed that using a pure ``AND'' \en policy, compared to ``OR'' or ``$k\text{-OutOf-}Q$'' policies, significantly increases the \en delay by increasing the number of organizations.
Similarly, \cite{end9} showed the same results by modeling HF using stochastic reward networks. They also observed that for ``OR'' and ``$k\text{-OutOf-}Q$'' policies the latency decreases by increasing the number of EPs within the same organization, similar to the effect of increasing $R$ in our work.

\cite{end6} modeled HF using Generalized Stochastic Petri Nets~(GSPN) and showed that for high request arrival rates, the \en phase is a performance bottleneck of HF. This is coherent with the motivation of our work, focusing on optimizing the \en phase. 
\cite{end2} considered four organizations and showed that simple \en policies based on ``AND'', ``OR'' and ``$k\text{-OutOf-}Q$'' operators, experience the minimum latency. 
\cite{end4} showed that using ``$k\text{-OutOf-}Q$'' policy, increasing $k$ decreases the throughput and increases the latency. This is coherent with our system model since the \en latency will be the maximum among $k$ request delays. 
\cite{end5} optimized the HF configurations to improve the throughput and reduce the delays. Coherently with our results, they showed the equivalence between the ``$1\text{-OutOf-}Q$'' policy and the ``OR'' among all organizations. Our results in Sec.~\ref{sec:endo} generalize such property.

Some works tried to improve \en phase of HF. \cite{soa} proposed a way to select the best EP for ``$1\text{-OutOf-}Q$'' endorsement policy in HF~v$1.4$. They introduced an algorithm running in each EP, called DSLM, to calculate the EP's load by considering multiple resource metrics within an EP. 
For each request, only half of the EPs are probed to get their actual load, coherently with $R=Q/2$ adopted in OPEN.
A version of DSLM tailored to our system model has been considered in Sec.~\ref{sec:simu} as an alternative approach to be compared with OPEN.
\cite{end8} showed that the failed transactions due to timeouts are affected by the number of statements within the ``AND'' operator defined in the \en policy. Such failures increase the latency and waste of resources due to re-transmissions at the application level.

 \cite{end11} suggested a way to reduce the possibility of endorsing conflicting transactions. They proposed a cache mechanism inside the EPs to record some data of the recently endorsed transactions and drop the conflicting proposal before execution. Recall that, in the endorsement phase, no execution results will update the world state, so transactions with similar initial world states can propose different updates for the world state. 
 This early drop of the proposal before execution will reduce the computing and network resources by reducing the chance of transaction failure at the validation phase. 
\cite{end1} removed unnecessary operations for pure read requests, by modifying the EPs algorithm to differentiate the process of pure read transactions from mixed read/write ones.
This reduced the latency and resource consumption in the endorsement phase.

The main idea of OPEN is to send multiple replicas of the same request to multiple peers. This approach has been deeply investigated in the literature on queueing theory, motivated by the problem of optimal job assignment to servers. As the literature is huge, we focus just on a few papers for the sake of space. In the generic literature about distributed systems, several works~\cite{red00, red0,red1} investigated the effect of sending replicas of a job to more than one randomly selected server and waiting for the first response to exploit redundancy, as in OPEN. These works introduced redundancy to reduce the job completion time and overcome server-side variability, where a server might be temporarily slow, due to many factors like garbage collection, background load, or even network interrupts. \cite{red2} showed that, besides its simplicity, in many cases, redundancy outperforms other techniques for overall response time. \cite{red3}, by decoupling the inherent job size from the server-side slowdown, described a more realistic model of redundancy and showed that increasing the level of redundancy can degrade the performance, coherently with our observations in Sec.~\ref{sec:mod}.
\cite{opt0} showed that a major improvement results from having each job replicated to only two servers, coherently with our Fig.~\ref{fig:Opt-R} which shows that for the 1-OutOf-$k$ policy, the \en latency decreases mostly when varying $R$ from 1 to 2. On the contrary, in our work, we have considered the optimal value of $R$ that minimizes the \en latency, which may be greater than 2.
 \cite{opt00} showed the reverse relation between the incoming load and the optimal number of replicas, coherently with~\eqref{eq:1rhat3}, and experimentally obtained the optimal redundancy factor in different job arrival rates and for different service times. Also, \cite{opt1} theoretically demonstrated that, when replicating the job to multiple servers, the best choice in case of low (or, high) loads is to replicate to all (or, only $1$) servers, coherently with~\eqref{eq:1rhat3} and with the operations of OPEN, which adapts the replication factor to the instantaneous load.

%% file: Sec7_Conclusions.tex
\section{Conclusions}\label{sec:conc}

We addressed the problem of minimizing the \en latency in HF. 
Leveraging some results obtained in a simplified queueing model, we proposed the {OPEN} algorithm to choose multiple EPs for each transaction by taking into account the measurements from the past requests, in a realistic scenario. Through simulations with OMNeT++, we showed that independently from the scenario, {OPEN} is robust and achieves performance remarkably close to the optimal oracle-based approach ({OOD}) and outperforms state-of-the-art solutions.

\tagged{trp}{
Due to the key role of \en policies, we expect that our results will inspire new research directions and implementation efforts in optimizing the performance of HF and other blockchain platforms. Given the complexity of the addressed problem, new solutions based on machine learning and meta-heuristics could be devised.}

OPEN has been validated only by extensive simulations. Beyond the scope of this work, we implemented OPEN in HF to validate the proposed approach in a realistic setting. The experimental results of the first version of the proof-of-concept are very promising. We leave the optimization of the design of the client-based OPEN solution and its extensive experimental validation for future work.

%% file: main.bbl
\begin{thebibliography}{10}
\providecommand{\url}[1]{#1}
\csname url@samestyle\endcsname
\providecommand{\newblock}{\relax}
\providecommand{\bibinfo}[2]{#2}
\providecommand{\BIBentrySTDinterwordspacing}{\spaceskip=0pt\relax}
\providecommand{\BIBentryALTinterwordstretchfactor}{4}
\providecommand{\BIBentryALTinterwordspacing}{\spaceskip=\fontdimen2\font plus
\BIBentryALTinterwordstretchfactor\fontdimen3\font minus
  \fontdimen4\font\relax}
\providecommand{\BIBforeignlanguage}[2]{{%
\expandafter\ifx\csname l@#1\endcsname\relax
\typeout{** WARNING: IEEEtran.bst: No hyphenation pattern has been}%
\typeout{** loaded for the language `#1'. Using the pattern for}%
\typeout{** the default language instead.}%
\else
\language=\csname l@#1\endcsname
\fi
#2}}
\providecommand{\BIBdecl}{\relax}
\BIBdecl

\bibitem{bitcoin}
S.~Nakamoto, ``Bitcoin: A peer-to-peer electronic cash system,''
  \emph{Decentralized Business Review}, 2008.

\bibitem{eth}
G.~Wood, ``Ethereum: a secure decentralised generalised transaction ledger,''
  \emph{Ethereum Project Yellow Paper}, 2014.

\bibitem{bc1}
L.~Hughes, Y.~K. Dwivedi, S.~K. Misra, N.~P. Rana, V.~Raghavan, and V.~Akella,
  ``Blockchain research, practice and policy: Applications, benefits,
  limitations, emerging research themes and research agenda,'' \emph{Journal of
  Information Management}, 2019.

\bibitem{bc2}
S.~F. Wamba and M.~M. Queiroz, ``Blockchain in the operations and supply chain
  management: Benefits, challenges and future research opportunities,''
  \emph{International Journal of Information Management}, 2020.

\bibitem{bc3}
I.~Abu-elezz, A.~Hassan, A.~Nazeemudeen, M.~Househ, and A.~Abd-alrazaq, ``The
  benefits and threats of blockchain technology in healthcare: A scoping
  review,'' \emph{International Journal of Medical Informatics}, 2020.

\bibitem{hf}
S.~Solat, P.~Calvez, and F.~Naït-Abdesselam, ``Permissioned vs. permissionless
  blockchain: How and why there is only one right choice,'' \emph{Journal of
  Software}, 2020.

\bibitem{hf1}
E.~Androulaki, A.~Barger, V.~Bortnikov, C.~Cachin, K.~Christidis, A.~De~Caro,
  D.~Enyeart, C.~Ferris, G.~Laventman, Y.~Manevich \emph{et~al.},
  ``{Hyperledger Fabric}: a distributed operating system for permissioned
  blockchains,'' in \emph{ACM EuroSys}, 2018.

\bibitem{corda}
\BIBentryALTinterwordspacing
R.~G. Brown, J.~Carlyle, I.~Grigg, and M.~Hearn, ``Corda: An introduction,''
  2016. [Online]. Available:
  \url{https://docs.r3.com/en/pdf/corda-introductory-whitepaper.pdf}
\BIBentrySTDinterwordspacing

\bibitem{hfend}
\BIBentryALTinterwordspacing
``Endorsement policies.'' [Online]. Available:
  \url{https://hyperledger-fabric.readthedocs.io/en/release-2.2/endorsement-policies.html}
\BIBentrySTDinterwordspacing

\bibitem{end6}
P.~Yuan, K.~Zheng, X.~Xiong, K.~Zhang, and L.~Lei, ``Performance modeling and
  analysis of a {Hyperledger}-based system using {GSPN},'' \emph{Computer
  Communications, Elsevier}, 2020.

\bibitem{end3}
X.~Xu, G.~Sun, L.~Luo, H.~Cao, H.~Yu, and A.~V. Vasilakos, ``Latency
  performance modeling and analysis for {Hyperledger Fabric} blockchain
  network,'' \emph{Information Processing \& Management}, 2021.

\bibitem{mm1}
P.~Purdue, ``The {M/M/1} queue in a {Markovian} environment,'' \emph{Operations
  research}, vol.~22, no.~3, pp. 562--569, 1974.

\bibitem{order}
H.~A. David and H.~N. Nagaraja, \emph{Order statistics}.\hskip 1em plus 0.5em
  minus 0.4em\relax John Wiley \& Sons, 2004.

\bibitem{omnet}
\BIBentryALTinterwordspacing
``Omnet++.'' [Online]. Available: \url{https://omnetpp.org/}
\BIBentrySTDinterwordspacing

\bibitem{intTopoZoo}
``{The Internet Topology Zoo},'' \url{http://www.topology-zoo.org}.

\bibitem{soa}
C.~Liu, M.~Li, Y.~Wang, Y.~Wang, D.~Huo, and Y.~Chen, ``Achieve better
  endorsement balance on blockchain systems,'' in \emph{CSCWD}.\hskip 1em plus
  0.5em minus 0.4em\relax IEEE, 2021.

\bibitem{end9}
H.~Sukhwani, N.~Wang, K.~S. Trivedi, and A.~Rindos, ``Performance modeling of
  {Hyperledger Fabric} (permissioned blockchain network),'' in
  \emph{NCA}.\hskip 1em plus 0.5em minus 0.4em\relax IEEE, 2018.

\bibitem{end2}
P.~Thakkar, S.~Nathan, and B.~Viswanathan, ``Performance benchmarking and
  optimizing {Hyperledger Fabric} blockchain platform,'' in
  \emph{MASCOTS}.\hskip 1em plus 0.5em minus 0.4em\relax IEEE, 2018.

\bibitem{end4}
L.~Hang and D.-H. Kim, ``Optimal blockchain network construction methodology
  based on analysis of configurable components for enhancing {Hyperledger
  Fabric} performance,'' \emph{Blockchain: Research and Applications}, 2021.

\bibitem{end5}
S.~Figueroa-Lorenzo, J.~A{\~n}orga, and S.~Arrizabalaga, ``Methodological
  performance analysis applied to a novel {IIoT} access control system based on
  permissioned blockchain,'' \emph{Information Processing \& Management}, 2021.

\bibitem{end8}
J.~A. Chacko, R.~Mayer, and H.-A. Jacobsen, ``Why do my blockchain transactions
  fail? a study of hyperledger fabric,'' in \emph{Proceedings of the 2021
  International Conference on Management of Data}, ser. SIGMOD '21.\hskip 1em
  plus 0.5em minus 0.4em\relax Association for Computing Machinery, 2021.

\bibitem{end11}
F.~Lu, L.~Gan, Z.~Dong, W.~Li, H.~Jin, and A.~Y. Zomaya, ``A cache enhanced
  endorser design for mitigating performance degradation in {Hyperledger
  Fabric},'' in \emph{IEEE International Conference on Blockchain}, 2018.

\bibitem{end1}
M.~Kwon and H.~Yu, ``Performance improvement of ordering and endorsement phase
  in {Hyperledger Fabric},'' in \emph{IOTSMS}.\hskip 1em plus 0.5em minus
  0.4em\relax IEEE, 2019.

\bibitem{red00}
N.~B. Shah, K.~Lee, and K.~Ramchandran, ``When do redundant requests reduce
  latency?'' \emph{IEEE Transactions on Communications}, 2015.

\bibitem{red0}
A.~Vulimiri, O.~Michel, P.~B. Godfrey, and S.~Shenker, ``More is less: Reducing
  latency via redundancy,'' in \emph{HotNets}.\hskip 1em plus 0.5em minus
  0.4em\relax ACM, 2012.

\bibitem{red1}
A.~Vulimiri, P.~B. Godfrey, R.~Mittal, J.~Sherry, S.~Ratnasamy, and S.~Shenker,
  ``Low latency via redundancy,'' in \emph{CoNEXT}.\hskip 1em plus 0.5em minus
  0.4em\relax ACM, 2013.

\bibitem{red2}
K.~Gardner, S.~Zbarsky, S.~Doroudi, M.~Harchol-Balter, and E.~Hyytia,
  ``Reducing latency via redundant requests: Exact analysis,'' \emph{ACM
  SIGMETRICS Performance Evaluation Review}, 2015.

\bibitem{red3}
K.~Gardner, M.~Harchol-Balter, A.~Scheller-Wolf, and B.~Van~Houdt, ``A better
  model for job redundancy: Decoupling server slowdown and job size,''
  \emph{IEEE/ACM Transactions on Networking}, vol.~25, no.~6, pp. 3353--3367,
  2017.

\bibitem{opt0}
K.~Gardner and S.~Zbarsky, ``Analyzing response time in the redundancy-d
  system,'' \emph{Poster, SIGMETRICS, Technical Report CMU-C S-15-141SIGMETRICS
  Performance Evaluation Review}, 2015.

\bibitem{opt00}
J.~Hollinghurst, A.~Ganesh, and T.~Baugé, ``Latency reduction in communication
  networks using redundant messages,'' in \emph{ITC}, 2017.

\bibitem{opt1}
G.~Joshi, E.~Soljanin, and G.~Wornell, ``Efficient replication of queued tasks
  for latency reduction in cloud systems,'' in \emph{Allerton Conference on
  Communication, Control, and Computing}, 2015.

\end{thebibliography}
